\DocumentMetadata{}
\documentclass{iacrtrans}
\definecolor{TUMBlau}{RGB}{0,101,189} % Pantone 300 (0, 0.396, 0.7412)
\definecolor{TUMBlauDunkel}{RGB}{0,82,147} % Pantone 301 (0, 0.3216,0.5765)
\definecolor{TUMBlauHell}{RGB}{152,198,234} % Pantone 283 (0.596,0.776,0.918)
\definecolor{TUMBlauMittel}{RGB}{100,160,200} % Pantone 542 (0.392, 0.627, 0.784)

\definecolor{TUMElfenbein}{RGB}{218,215,203} % Pantone 7527 -Elfenbein
\definecolor{TUMGruen}{RGB}{162,173,0} % Pantone 383 - Grün (0.6353,0.6784,0)
\definecolor{TUMOrange}{RGB}{227,114,34} % Pantone 158 - Orange (0.8902, 0.4471, 0.133)
\definecolor{TUMGrau}{gray}{0.6} % Grau 60%

\definecolor{TUMGruenDunkel}{RGB}{0,124,48} % (0,0.4863,0.1882)
\definecolor{TUMRot}{RGB}{196,7,27} % (0.7686,0.02745,0.10588)

\definecolor{TUMBlue}{RGB}{0,101,189} % Pantone 300 (0, 0.396, 0.7412)
\definecolor{TUMBlueDark}{RGB}{0,82,147} % Pantone 301 (0, 0.3216,0.5765)
\definecolor{TUMBlueLight}{RGB}{152,198,234} % Pantone 283 (0.596,0.776,0.918)
\definecolor{TUMBlueMedium}{RGB}{100,160,200} % Pantone 542 (0.392, 0.627, 0.784)

\definecolor{TUMIvory}{RGB}{218,215,203} % Pantone 7527 -Elfenbein
\definecolor{TUMGreen}{RGB}{162,173,0} % Pantone 383 - Grün (0.6353,0.6784,0)
\definecolor{TUMGray}{gray}{0.6} % Grau 60%
\definecolor{TUMGrayDark}{gray}{0.3} % Grau 80%

\definecolor{TUMGreenDark}{RGB}{0,124,48} % (0,0.4863,0.1882)
\definecolor{TUMRed}{RGB}{196,7,27} % (0.7686,0.02745,0.10588)

\definecolor{plotColor1}{RGB}{0,101,189} % TUMBlue, Pantone 300 (0, 0.396, 0.7412)
\definecolor{plotColor2}{RGB}{0,124,48} % TUMGreenDark, (0,0.4863,0.1882)
\definecolor{plotColor3}{RGB}{196,7,27} % TUMRed, (0.7686,0.02745,0.10588)
\definecolor{plotColor4}{RGB}{227,114,34} % TUMOrange, Pantone 158 - Orange (0.8902, 0.4471, 0.133)
\definecolor{plotColor5}{RGB}{0,82,147} % TUMBlueDark, Pantone 7527 -Elfenbein
\definecolor{plotColor6}{RGB}{162,173,0} % TUMGreen, Pantone 283 (0.596,0.776,0.918)
\definecolor{plotColor7}{gray}{0.3} % TUMGrayDark, Grau 80%

\definecolor{plotColorIA1}{RGB}{0,124,48} % TUMGreenDark, (0,0.4863,0.1882)
\definecolor{plotColorIA2}{RGB}{0,82,147} % TUMBlueDark, Pantone 7527 -Elfenbein
\definecolor{plotColorIA3}{RGB}{100,160,200} % TUMBlueMedium, Pantone 542 (0.392, 0.627, 0.784)
\definecolor{plotColorIA4}{RGB}{152,198,234} % TUMBlueLight, Pantone 283 (0.596,0.776,0.918)
\definecolor{plotColorIA5}{gray}{0.3} % TUMGrayDark, Grau 80%
\definecolor{plotColorIA6}{RGB}{227,114,34} % TUMOrange, Pantone 158 - Orange (0.8902, 0.4471, 0.133)
\definecolor{plotColorIA7}{RGB}{196,7,27} % TUMRed, (0.7686,0.02745,0.10588)

\newif\ifshowComment
 \showCommentfalse

\usepackage{amsmath}
\usepackage{amssymb,mathtools,amsbsy}
\usepackage{bbm}
\usepackage{epsfig}
\usepackage{mathdots}
\usepackage{booktabs}
\usepackage{tabularx}
\usepackage{longtable} % Table with pagebreaks, must be before "arydshln"

\usepackage{array}
\usepackage{multirow}
\usepackage{arydshln}
\newcolumntype{C}{>{$}c<{$}} % math-mode version of "c" column type

\usepackage[hidelinks]{hyperref}

\usepackage[noend]{algpseudocode}
\usepackage[linesnumbered,ruled,vlined,titlenumbered]{algorithm2e}
\usepackage{algorithmicx}

\usepackage{pgfplots}
\pgfplotsset{compat=newest}
\usepgfplotslibrary{fillbetween}
\usepackage{tikz}
\usetikzlibrary{matrix,arrows,arrows.meta,positioning,calc,shapes}
\usepackage{circuitikz}

\usepackage{caption,subcaption}

\usepackage{enumitem}

\usepackage{stackengine} % Two-lined text

\usepackage{footnote}
\makeatletter
\newcommand{\footnoteref}[1]{\protected@xdef\@thefnmark{\ref{#1}}\@footnotemark}
\makeatother

\usepackage{nicefrac} % nicer in-line fractions

\usepackage{amsthm}
\usepackage{thm-restate}

\usepackage[capitalise]{cleveref}
\newcommand{\cD}{\mathcal{D}}

\newcommand{\cK}{\mathcal{K}}
\newcommand{\cL}{\mathcal{L}}
\newcommand{\cM}{\mathcal{M}}
\newcommand{\cN}{\mathcal{N}}

\newcommand{\cR}{\mathcal{R}}
\newcommand{\cS}{\mathcal{S}}

\newcommand{\cW}{\mathcal{W}}
\newcommand{\cX}{\mathcal{X}}
\newcommand{\cY}{\mathcal{Y}}
\newcommand{\cZ}{\mathcal{Z}}

\DeclareMathOperator{\supp}{supp}

\newcommand*\diff{\mathop{}\!\mathrm{d}}

\tikzset{%
	dots/.style args={#1per #2}{%
		line cap=round,
		dash pattern=on 0 off #2/#1
	}
}
\definecolor{light-gray}{gray}{0.7}
\definecolor{dark-gray}{gray}{0.3}
\usetikzlibrary{spy}
\usetikzlibrary{backgrounds}
\usetikzlibrary{decorations}
\usetikzlibrary{patterns}
\usetikzlibrary{arrows,matrix,positioning,calc}

\renewcommand{\epsilon}{\varepsilon}
\definecolor{darkgray}{RGB}{64,64,64}
\definecolor{litegray}{RGB}{192,192,192}
\tikzstyle{block}=[draw, rectangle, minimum height=1cm, text width=1.5cm, text centered, draw=darkgray, font=\small]
\tikzstyle{block_medium}=[draw, rectangle, minimum height=1.5cm, text width=2cm, text centered, draw=darkgray, font=\small]
\tikzstyle{block_large}=[draw, rectangle, minimum height=1.75cm, text width=3cm, text centered, draw=darkgray, font=\small]
\tikzstyle{line} = [draw, -latex]

\usepackage{bm}
\definecolor{mycolor1}{rgb}{0.00000,0.44700,0.74100}
\definecolor{mycolor2}{rgb}{0.85000,0.32500,0.09800}
\definecolor{mycolor3}{rgb}{0.4660, 0.6740, 0.1880}

\def\qed{\hfill $\blacksquare$}
\usepackage{color}

\usepackage{stackengine}

\renewcommand{\epsilon}{\varepsilon}

\tikzstyle{line} = [draw, -latex]

\usepackage{mathtools}
\newcommand{\defeq}{\vcentcolon=}

\def\qed{\hfill $\blacksquare$}

\tikzstyle{block_medium}=[draw, rectangle, minimum height=1.5cm, text width=2cm, text centered, draw=darkgray, font=\small]
\tikzstyle{block_large}=[draw, rectangle, minimum height=2cm, text width=2cm, text centered, draw=darkgray, font=\small]

\usepackage[ruled,linesnumbered]{algorithm2e}
\usetikzlibrary{shapes.geometric,arrows,positioning}
\tikzstyle{decision} = [diamond, draw, fill=blue!20, 
text width=5em, text badly centered, node distance=4cm, inner sep=0pt]
\tikzstyle{block} = [rectangle, draw, fill=blue!20,  text centered, rounded corners, minimum height=4em]
\tikzstyle{line} = [draw, -latex']
\tikzstyle{cloud} = [draw, ellipse,fill=red!20, node distance=6.6cm,
minimum height=1em]
\tikzstyle{algorithm} = [rectangle, draw, fill=green!20,  text centered, rounded corners, minimum height=4em, minimum width =6em]
\tikzstyle{initialization} = [rectangle, draw,   text centered, minimum height=4em, minimum width =6em]

\usetikzlibrary{spy}
\usetikzlibrary{backgrounds}
\usetikzlibrary{decorations}
\usetikzlibrary{patterns}
\usetikzlibrary{arrows,matrix,positioning,calc}

\DeclareMathOperator{\erfc}{erfc}

\newcommand*\fooA{\mathrel{-\mkern-3mu{\circ}\mkern-3mu-}}

\usepackage{appendix}
\AtBeginEnvironment{subappendices}{%
	\chapter*{Appendix}
	\addcontentsline{toc}{chapter}{Appendices}
	\counterwithin{figure}{section}
	\counterwithin{table}{section}
}
\usepackage[utf8]{inputenc} 
\usepackage[T1]{fontenc}
\usepackage{url}
\usepackage{ifthen}
\usepackage{cite}

\usepackage{xcolor}

\usepackage{bm}
\usepackage{pgfplots}
\usepackage{pifont}
\usepackage{amsthm}

\usetikzlibrary{external}
\begin{document}
\title{Information Theoretic Analysis of \\PUF-Based Tamper Protection}

\author{Georg Maringer\inst{1,2} \and Matthias Hiller\inst{2}}
\institute{
    Technical University of Munich, Munich, Germany, \email{georg.maringer@tum.de}
    \and
    Fraunhofer AISEC, Garching, Germany, \email{{georg.johannes.maringer, matthias.hiller}@aisec.fraunhofer.de}
}

\maketitle


\begin{abstract}
   Physical Unclonable Functions (PUFs) enable physical tamper protection for high-assurance devices without needing a continuous power supply that is active over the entire lifetime of the device. Several methods for PUF-based tamper protection have been proposed together with practical quantization and error correction schemes. In this work we take a step back from the implementation to analyze theoretical properties and limits. We apply zero leakage output quantization to existing quantization schemes and minimize the reconstruction error probability under zero leakage. We apply wiretap coding within a helper data algorithm to enable a reliable key reconstruction for the legitimate user while  guaranteeing a selectable reconstruction complexity for an attacker, analogously to the security level for a cryptographic algorithm for the attacker models considered in this work. We present lower bounds on the achievable key rates depending on the attacker's capabilities in the asymptotic and finite blocklength regime to give fundamental security guarantees even if the attacker gets partial information about the PUF response and the helper data. Furthermore, we present converse bounds on the number of PUF cells. Our results show for example that for a practical scenario one needs at least 459 PUF cells using 3 bit quantization to achieve a security level of 128 bit.
\end{abstract}

\section{Introduction}

Physical Unclonable Functions (PUFs) evaluate physical properties of devices to obtain unique identifiers of electronic devices and provide physical roots of trust for cryptographic keys. Furthermore, PUFs can serve as a foundation for tamper protection technology that facilitates to validate the physical integrity of an embedded system after its power-up. All approaches have in common that minuscule manufacturing variations within physical objects, or mostly electronic components, are evaluated to generate an internal device-unique output. While there are several works on the assessment of the entropy or randomness of PUFs \cite{maiti2013systematic,wilde2018spatial,frisch2023practical} and the leakage through the published helper data necessary within the reconstruction phase, e.g. \cite{DGSV15,DGV+16}, we are currently lacking a theoretical model to quantify the security in the light of a physical attacker who destroys parts of the PUF response to read out the remainder.

Going from silicon PUFs, e.g. SRAM, ring oscillator or arbiter PUF \cite{HYKD14}, to system-level PUFs facilitates to incorporate tamper protection capabilities to protect an entire embedded device with components that cannot resist advanced physical attacks on their own, such as processors, FPGAs, or external memories and their communication, or discrete components that are susceptible e.g. to side-channel attacks. This reduces the attack surface from several individual vulnerabilities, e.g., against laser or EM fault injection, EM or optical side channel analysis, and analysis of digital communication interfaces between components to the attack resistance of the surrounding barrier. We also consider it infeasible to perform power side-channels since the attacker can only access the power supply of the entire printed circuit board and has no direct access to the power supply of the chip or individual discrete components.

For the remainder of this work, we will refer to the PUF-based tamper protection foil proposed by Immler \emph{et al.} \cite{IOK+18,ION+19}. However, the generic results of this work can be adapted and applied to other PUF types as well. Examples are the coating PUF \cite{tuyls2006read} and the polymer waveguide PUF \cite{vai2016secure,geis2017}. The tamper protection is based on a foil that is wrapped around an entire Printed Circuit Board (PCB), or a cover that is placed on top and bottom. This foil consists of a mesh of electrodes, leading to a large number of capacitances that can be measured by a mixed-signal circuit from within the protected area. It evaluates the capacitive coupling between electrodes to derive the cryptographic key and to validate the physical integrity of the system, and performs run-time tamper detection during operation to protect the system.

One of the critical attack vectors, considered during the evaluation of hardware devices with security boxes  \cite{jilHDSB}, is that an attacker penetrates the foil with a small needle or drill and accesses internal signals. If the required drill diameter is sufficiently large, major changes occur in the capacitance measurements  of a significant portion of the foil, leading to an incorrectly reconstructed PUF response during the reconstruction phase. Therefore, the secret cannot be uncovered by an attacker, as discussed, e.g., in \cite{GXKF22}. As the foil's PUF values may also change over time due to noise, aging, and varying environmental conditions such as temperature or humidity, an error-correcting code is implemented in the system to compensate for those effects to ensure that the correct cryptographic key is derived with a probability $> 1-10^{-6}$ or even $>1-10^{-9}$ so that the PUF does not have significant impact on the reliability of the overall system.

Our goal is to analyze the resulting wiretap channel \cite{wyner1975wire,csiszar1978broadcast} between the enrollment and reconstruction phase of the legitimate user as well as the reconstruction phase of the attacker from an information theoretical point of view. We establish lower bounds on the secrecy capacities of the resulting channels as well as finite blocklength achievability and converse bounds on the maximal achievable secrecy rate, making our results relevant in practice as they provide benchmarks for implementations by quantifying the distance of a practical implementation to the theoretical limit.

\subsection{Related Works}
For a survey on information and coding theoretic techniques covering enrollment and reconstruction without tamper protection see \cite{gunlu2020optimality}. We also mention literature in the context of biometric secrecy as this field is closely related to PUFs. \cite{gunlu2019code} for example code constructions for both biometric secrecy systems as well as PUFs are given. Achievable rate regions of biometric secrecy systems under security and privacy constraints are presented in \cite{ignatenko2009biometric}. Approaches to achieve biometric secrecy using Slepian-Wolf distributed source coding techniques are presented in \cite{vetro2009securing,draper2007using}.

\subsection{Main Results}

The main results presented in this work are:

\begin{itemize}
	\item Information theoretical channel model including zero leakage helper data generation for physical tampering with PUFs
	\item Asymptotic results for lower bounds on the channel capacity of the resulting PUF-channel under different attack scenarios
	\item Finite blocklength achievability and converse results on the number of required capacitances to achieve a predefined security level in two attack scenarios
	\item Proof that previously used helper data schemes do not achieve required security levels without leaking information about the secret via the helper data
	\item Quantitative results that demonstrate that a $128$ bit security level is achievable with $1400$ PUF cells for $18\%$ and $36\%$ erasure probability for digital and analog attacker, respectively
        \item Proof that an existing converse bound on finite blocklength wiretap codes cannot be tight for all channels
\end{itemize}

\subsection{Outline}

Section~\ref{sec:sota} gives a brief overview over related work. Section~\ref{sec:preliminaries} introduces the notation used throughout this work and recaps known results in the field of information theory, in particular for finite blocklength that are used to proof the main results presented in this work. Furthermore, helper data algorithms with an emphasis on zero leakage helper data are recapped and their connection between secret sharing using common randomness is examined. Section~\ref{sec:channel_model} gives some background information on the foil PUF and introduces the resulting channel model. In Section~\ref{sec:securing_foil}, we obtain results on the secret key capacity of the HDA for digital and analog attacker. Section~\ref{sec:secret_sharing_oneway} investigates secret sharing using common randomness using one-way communication for finite blocklength. It serves as a foundation to analyze the HDA performance with respect to the required amount of capacitive PUF cells for the foil PUF presented in Section~\ref{sec:finite}. In Section~\ref{sec:converse_application} we use the converse result on the necessary amount of PUF cells to show that a given implementation has to either leak via the helper data or be insecure by other means. Section~\ref{sec:conclusion_pufs} sums up the results and states open problems.

\section{State of the Art}\label{sec:sota}

While some silicon PUFs such as the SRAM or arbiter PUF directly output digital information, other silicon PUFs, like the Ring-Oscillator or TERO PUF, and in particular non-silicon PUFs output analog, or finely quantized digital data. They all have in common, that they undergo a processing chain involving helper data and error correction, until a cryptographic key is output (see Section~\ref{subsec:HDA}).

\subsection{PUF-Based Tamper Protection}\label{subsec:puf-basedtampprot}

In the past, tamper protection was implemented through a continuously powered detection mechanism, e.g. in tamper-responsive envelopes and covers that wrap or cover the structure to be protected \cite{IMFC13,OI18b}. Within the protective structure, a physical measure such as electrical resistance is measured and triggers an alarm as soon as the measured value exceeds a threshold to erase sensitive information, stored e.g. in battery-backed memory. While the device is only active during a fraction of the time, the protective measures need to be active for the entire life-time of the device, after it leaves a trusted manufacturing site. PUF-based tamper protection promises to increase sensitivity and to ease the handling during operation of the devices, as the device can be fully powered off. Thus PUFs, can contribute to an easy-to-handle and future-proof tamper protection technology.

Over the last 20 years, different measures were taken to combine PUFs and tamper protection. An early type is the coating PUF, where a protective coating is spread over an IC and evaluated from its inside to detect physical tampering when the coating is penetrated \cite{SMKT06,TSS+06}. In addition to electrical measurements, also optical approaches \cite{esbach2012new,vai2015secure}, mechanical properties \cite{GS22} or propagation behavior of radio waves \cite{STZP22} were proposed.

In this work, we use PUF-based tamper protection realized by foils and covers, as also proposed in \cite{IOK+18,Imm19,ION+19,GOFK21} as reference. This scenario is depicted in Figure~\ref{fig:pufenv}. The foil is wrapped around a PCB with mounted electronics components to be protected and consists of two structured electrode layers (blue and red) and two electrical shielding layers (black). The capacitive coupling between the electrode layers is measured to obtain the PUF response \cite{OIHS18}. In addition, faster mechanisms for run-time tamper protection are included.

\begin{figure}
	\centering
	\includegraphics[width=0.5\linewidth]{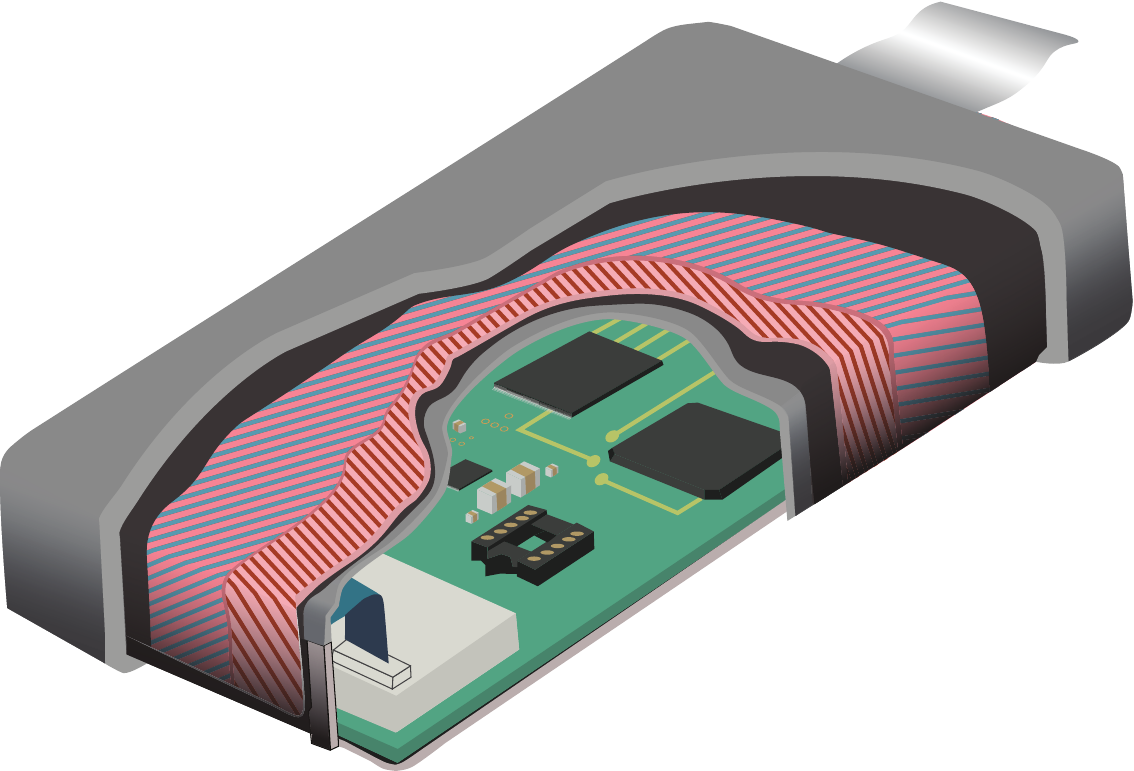}
	\caption{Sketch of PUF-based envelope \cite{IOK+18} with two conducting electrode layers in red and blue, and two shielding layers in black}
	\label{fig:pufenv}
\end{figure}

The capacitance values are subject to measurement noise, temperature and other environmental effects, as discussed based on measured values, e.g. in \cite{GXKF22}. A broad range of deterministic effects can be compensated with linear or also higher order reference points of fits \cite{OIHS18,GXKF22,riehm2023structured,gunlu2015reliable}, whereas Gaussian measurement noise remains in all cases. We will focus on this fundamental noise issue using a wiretap scenario in this work and refer to the compound case for including multiple environmental conditions \cite{liang2009compound,boche2013secret}.

\subsection{Quantization}

Analog PUFs evaluated on embedded devices need to quantize the digitized PUF data into a finite alphabet that is input into the error correcting code (see Section~\ref{subsec:hda_analog}). Typically, public helper data that references the distance and direction to the next interval border is stored to move a measured data point away from decision borders right into the middle of the quantization interval. So far, work on PUFs typically either uses equiprobable or equidistant quantization \cite{IHKS16}. As shown in Section~\ref{subsec:zero_leakage_hds} it is possible though to use arbitrary input quantization while still generating helper data that is not leaking any information about the secret (zero-leakage helper data).

\paragraph{Equiprobable Quantization} The range of possible output values of the PUF is split into $d$ intervals with the same probability such that all indices as sampled with the same probability. This is favorable from a security point of view, as iid PUF values are mapped into uniform data in the finite alphabet. However, this comes with two downsides: 

First, the common helper data generation bringing the expected value during reconstruction into the center of the interval leaks information about the secret, as the helper data pointers of different intervals have different distributions \cite{IHKS16}. Later, we will introduce a method to obtain zero leakage helper data such that this problem is mitigated. 

Second, the decision borders in the center are rather narrow, so that the values in this intervals are subject to a higher error rate. The wide intervals in the tails of the Gaussian distribution considered in our model also have a decreased sensitivity for physical tampering.

Equiprobable quantization can be applied if the requirement for uniformity is more substantial than the requirement for tamper detection sensitivity \cite{GXKF22}.

\paragraph{Equidistant Quantization} In contrast, equidistant quantization samples the analog values by mapping them to intervals of the same size. This has the advantage that additive noise effects all values in the same way and error probabilities are constant with respect to the PUF values. Also, only a negligible leakage can be observed \cite{IHKS16} through the aforementioned helper data pointers. With the later on introduced zero leakage helper data algorithms, this is less of a factor. As a downside, this comes at the expense of a heavy bias as the intervals differ in probability. This can be mitigated through variable-length encoding at the expense of a limited selection of code classes for the later ECC \cite{IHL+18}.
As shown in \cite{BDH+10}, also higher-dimensional structures can be used for embedding secret data, which adds more degrees of freedom.

\subsection{Error Correction}

Aside from the helper data used to perform better output quantization during reconstruction (denoted by $W^n$ later in Section~\ref{subsec:HDA}) additional helper data $\widetilde{W}$ is generated to link the PUF response to a codeword of an error correcting code. 
This can be done either by linear schemes that generate the link through linear dependencies such as syndromes or XORs or pointers that refer to parts of the PUF response \cite{HKS20}.

In any case, an ECC is used to reduce the key error probability e.g. down to $10^{-6}$ or $10^{-9}$ to generate reliable keys. This can be achieved with standard codes such as BCH, Reed-Solomon or Convolutional codes \cite{MS77}, or newer code classes such as limited magnitude codes \cite{IU19} or Polar codes \cite{CIW+17,GXKF22}.

In addition, wiretap codes can be used either for leakage prevention \cite{HO17} or attack prevention \cite{GXKF22}.

Over the last decades, several schemes have been proposed and implemented to address the design of the error correcting codes (ECCs) and helper data, see e.g. \cite{DGSV15,HKS20}. In the following, we focus on the generation of input and output quantizers as well as fundamental theoretical limits on the amount of required PUF cells. Designers can then benchmark their implementations against the fundamental limits. The construction of a practical wiretap code and considerations for the leakage of the helper data connecting the PUF responses and the coding scheme are therefore out of the scope of this work.

\section{Preliminaries} \label{sec:preliminaries}
\subsection{Notation}
We denote scalars by lowercase letters and vectors by lowercase bold letters. Matrices are denoted by uppercase bold letters. We denote the $i$-th column of the matrix $\mathbf{A}$ by $\mathbf{a}_i$.

We denote random variables (RVs) by uppercase letters and their realizations by lowercase letters, i.e. the realization of a RV $X$ is denoted by $x$. Furthermore, we denote the probability mass function of a discrete RV $X$ by $P_X(x)$, the probability density function of a continuous RV $X$ by $f_X(x)$ and the cumulative distribution by $F_X(x)$ in both cases. We denote the expectation operator by $\mathbb{E}[\cdot]$. If three random variables $X,Y,Z$ form a Markov chain, i.e. $P_{XYZ}(x,y,z) = P_X(x) P_{Y|X}(y|x) P_{Z|Y}(z|y)$, we write $X \fooA Y \fooA Z$.

We denote the Gaussian distribution with mean $\mu$ and variance $\sigma^2$ by $\cN(\mu,\sigma^2)$ and we write that a RV $X$ is distributed according to a Gaussian by $X \sim \cN(\mu,\sigma^2)$.

We denote the Gaussian Q-function by
\begin{equation*}
	Q(x) \defeq \int_x^\infty \frac{1}{\sqrt{2\pi}} \exp\left(-\frac{z^2}{2}\right) \diff z
\end{equation*}
and the complementary error function by
\begin{equation*}
	\erfc(x) \defeq \frac{2}{\sqrt{\pi}} \int_x^\infty \exp(-z^2) \diff z \enspace .
\end{equation*}

We define for a functions $f(n),g(n)$ the Landau symbols 
\begin{equation*}
	f(n)\in o(g(n)) \Leftrightarrow \lim_{n \to \infty} \left\lvert \tfrac{f(n)}{g(n)}\right\rvert=0
\end{equation*}	
and 
\begin{equation*}
	 f(n) \in \mathcal{O}(g(n)) \Leftrightarrow  \limsup_{n \to \infty} \left\lvert\tfrac{f(n)}{g(n)}\right\rvert < \infty \enspace .
\end{equation*}

We denote sets by caligraphic letters, e.g., a set $\cS$ and its cardinality by $|\cS|$.

To keep the paper self-contained, we provide selected basics in information and coding theory used in the following sections in Appendix~\ref{app:basics_info_coding}.

\subsection{The Wiretap Channel}\label{subsec:wiretap_channel}
In \cite{wyner1975wire} Wyner introduced a channel model, in the following referred to as the Degraded Wiretap Channel (DWTC) (see also \cite{bloch2011physical}). The channel has one input, in the following denoted by the random variable $A$, and two outputs $B$ and $E$ accessed by Bob and Eve, respectively. This is illustrated in the upper part of Fig.~\ref{fig:WT_Channel}. The goal of Alice is to use this channel to reliably transmit a message $m$ to Bob, i.e. $m=\hat{m}$ with high probability while keeping any information about the message secret from an eavesdropper called Eve. For the degraded case it is obvious that the channel from Alice to Bob has higher channel capacity compared to the one from Alice to Eve. The joint distribution of the channel input and its output is given as $P_{ABE}(a,b,e) = P_A(a) P_{B|A}(b|a) P_{E|B}(e|b)$. The same does not necessarily hold for the more general wiretap channel (WTC) studied by Csiszar and Körner in \cite{csiszar1978broadcast} and shown in the lower part in Fig.~\ref{fig:WT_Channel}. Here, the eavesdropper's channel output $E$ is not a degraded version of $B$ rather it is directly generated from $A$ through a noisy channel, i.e. $P_{ABE}(a,b,e) = P_A(a) P_{B|A}(b|a) P_{E|A}(e|a)$ in this case.

\begin{figure}[t]
	\begin{center}
		\includegraphics{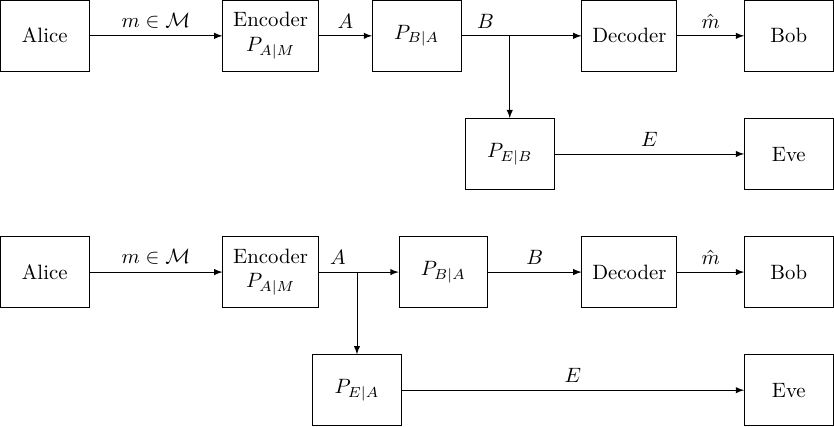}
	\end{center}
	\caption{Degraded wiretap channel (top), wiretap channel (bottom)}
	\label{fig:WT_Channel}
\end{figure}

The next definition is a slightly adapted version of the definition of secrecy codes in \cite{yang2019wiretap}. 
\begin{definition}\label{def:secrecy_coding_strategy}
	An $(|\mathcal{M}|,\varepsilon,\delta)$ secrecy coding strategy (also referred to as a wiretap code) for a (degraded) wiretap channel $(\mathcal{A},P_{B,E|A},\mathcal{B}\times \mathcal{E})$ consists of
	\begin{itemize}
		\item a set of possible messages $\mathcal{M}:=\{1,\ldots,|\mathcal{M}|\}$ from which a message $M=m$ is selected,
		\item a randomized encoder that generates a codeword $A(m)$ for $m \in \mathcal{M}$ according to a pdf $P_{A|M=m}$ and
		\item a decoder $Dec:\mathcal{B} \to \mathcal{M}$ that assigns an estimate $\hat{M}$ to each received signal $B \in \mathcal{B}$.
	\end{itemize}
	Encoder and decoder satisfy the average error probability (averaging performed by uniformly sampling the input message and over the randomness induced by the stochastic encoder)
	\begin{equation*}
		Pr(Dec(B) \neq M) \leq \varepsilon
	\end{equation*}
	where for $B$ it holds that $B$ is distributed according to $P_{B|M}(b|m) = \sum_{a \in \mathcal{A}} P_{B|A}(b|a) P_{A|M}(a|m)$. We distinguish average secrecy and maximum secrecy in the following way. For average secrecy with security parameter $\delta$ we require that
	\begin{equation}\label{eq:avg_security}
		d(P_{ME},P_M^{unif} P_E) \leq \delta \enspace .
	\end{equation}
	where $P_M^{unif}$ denotes the uniform distribution over the space of possible messages. In contrast for maximum secrecy it has to hold that
	\begin{equation}\label{eq:max_security}
		\max_{m\in\mathcal{M}} d(P_{E|M=m},Q_E) \leq \delta \enspace ,
	\end{equation}
	where $Q_E$ is the marginal distribution of $E$ if uniformly distributed messages are transmitted over the channel.
	If we make the number of channel uses specific, we call an $(|\mathcal{M}|,\varepsilon,\delta)_{avg}$ average secrecy code for the channel $P_{B^nE^n|A^n}$ an $(n,|\mathcal{M}|,\varepsilon,\delta)_{avg}$ secrecy code. Codes for maximum secrecy are denoted by $(n,|\mathcal{M}|,\varepsilon,\delta)_{max}$. In the following we frequently omit specifying whether we are interested in the average or maximum secrecy setting. The definitions work analogously for both in those cases and when it is not clear from context we specify average or maximum secrecy in the index. We define the maximal achievable secrecy rate by
	\begin{equation}
		R^*(n,\varepsilon,\delta):= \max\left\{\frac{\log(|\mathcal{M}|)}{n}: \exists (n,|\mathcal{M}|,\varepsilon,\delta) \text{ secrecy code}\right\} \enspace .
	\end{equation}
\end{definition}
It makes intuitively sense that the security conditions in equations~\eqref{eq:avg_security} and \eqref{eq:max_security} make it hard for an attacker to obtain information about the message $m$. The following Theorem quantifies this statement for average secrecy.
\begin{theorem}[\cite{yang2016finite} Thm. 8]\label{th:listsize}
	Let the output of an arbitrary list decoder $\cL$ given Eve's observation $E$ be denoted by $\cL(E)$. Let $\delta$ be the secrecy parameter of the implemented secrecy code for the respective wiretap channel. Then the probability that the transmitted codeword is not in the output list of Eve's list decoder having listsize $L$ is lower bounded by
	\begin{equation}
		P_{ME}(M \notin \mathcal{L}(E)) \geq 1 - \delta - \frac{L}{|\mathcal{M}|} \enspace .
	\end{equation}
\end{theorem}

\begin{definition}
	The secrecy capacity $C_S$ of a wiretap channel is defined by
	\begin{equation}
		C_S := \frac{1}{n} \limsup_{n\to \infty} R^*(n,\varepsilon,\delta)
	\end{equation}
	for arbitrarily small values $\varepsilon>0$ and $\delta>0$.
\end{definition}
The secrecy capacity for both channels is well known and given in the following Theorems.
\begin{remark}
	Notice that we did not distinguish the secrecy capacity for average and maximal secrecy because their value is the same. However for fixed $n,\varepsilon,\delta$, the values $R^*_{avg}(n,\varepsilon,\delta)$ and $R^*_{max}(n,\varepsilon,\delta)$ can be different.
\end{remark}

\begin{theorem}[\cite{wyner1975wire}]\label{th:cap_dwtc}
	For the degraded wiretap channel with input $A$ and outputs $B$ and $E$ for legitimate and eavesdropper, respectively, the secrecy capacity is equal to
	\begin{equation}\label{eq:degraded_wiretap_channel}
		C_S = \max_{P_A} I(A;B)-I(A;E) \enspace .
	\end{equation}
\end{theorem}
\begin{theorem}[\cite{csiszar1978broadcast}]\label{th:cap_wtc}
	For the wiretap channel with input $A$ and outputs $B$ and $E$ for legitimate and eavesdropper, respectively, the secrecy capacity is equal to
	\begin{equation}\label{eq:wiretap_channel}
		C_S = \max_{P_{V,A}} I(V;B)-I(V;E) \enspace ,
	\end{equation}
	where $V$ serves as an auxiliary random variable and it holds that $V \fooA A \fooA (B,E)$.
\end{theorem}
\begin{remark}
	Notice that in both cases the capacity does not depend on $\varepsilon$ and $\delta$. This changes for the task of determining $R^*(n,\varepsilon,\delta)$ for finite $n$. We introduce bounds on $R_{avg}^*(n,\varepsilon,\delta)$ and $R_{max}^*(n,\varepsilon,\delta)$ for finite $n$ later in Section~\ref{subsec:finite_blocklength_wiretap}.
\end{remark}

\subsection{Secret Sharing using Common Randomness}\label{subsec:secret_sharing}

\begin{figure}[t!]
	\begin{centering}
		\resizebox{\textwidth}{!}{
			\includegraphics{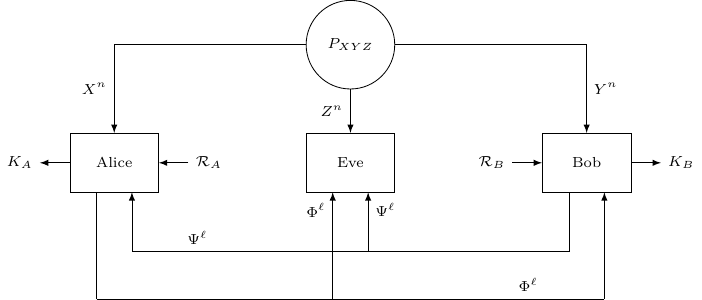}
		}
		\caption{Secret Sharing using Common Randomness}
		\label{fig:secret_sharing}
	\end{centering}
\end{figure}

Secret sharing using common randomness has been investigated by Maurer in \cite{maurer1993} and by Ahlswede and Csiszàr in \cite{ahlswede1993common}. The problem is graphically illustrated in Fig.~\ref{fig:secret_sharing}. In this section we explain the known results which will later be used in Section~\ref{sec:securing_foil}. In the following we describe the problem of deriving a secret key that is shared between two terminals, in the following called Alice and Bob, when the terminals have access to common randomness and are able to send messages to each other over a public channel. The source of common randomness is specified by a joint probability mass function $P_{XYZ}$. We denote the respective random variables specifying its outputs by $X,Y,Z$. This source is iid sampled $n$ times and we denote the random variables specifying the output of this process by $(X^n,Y^n,Z^n)$. The first terminal, in the following denoted by Alice, receives the sequence $X^n = (X_1,\ldots,X_n)$, while the second terminal (Bob) gets the sequence $Y^n = (Y_1,\ldots,Y_n)$. The third terminal (Eve), which is an adversary trying to obtain knowledge about the secret key that Alice and Bob shall agree on, is provided with the sequence $Z^n = (Z_1,\ldots,Z_n)$. The probability mass function $P_{XYZ}$ is publicly known, in particular by Alice, Bob and Eve. The goal of Alice and Bob is to reliably agree on a secret key while leaving Eve oblivious about it. To achieve this goal, Alice and Bob send messages to each other over the public channel that depend on their apriori knowledge of $P_{XYZ}$ and their respective shares $X^n$ or $Y^n$. Subsequent messages may also depend on previously received messages over the public channel coming from the other terminal. Eve is able to eavesdrop those messages but is unable to alter them or to insert additional messages into the public channel. We denote the $i$-th message sent from Alice to Bob by $\Phi_i$ and the $i$-th message sent from Bob to Alice by $\Psi_i$. Furthermore, we define $\Phi^i \defeq (\Phi_1,\ldots,\Phi_i)$ and $\Psi^i \defeq (\Psi_1,\ldots,\Psi_i)$. After Alice and Bob are finished with exchanging messages over the public channel, say after $\ell$ steps, Alice computes a key $K_A$ and Bob computes a key $K_B$ that are both within the same keyspace denoted by $\cK$. For the case that both keys are equal we simply denote the key by $K$.

As for transmitting data securely over a wiretap channel, it is essential for Alice and Bob to have access to local randomness to randomize the encoding function for the messages sent over the public channel. Hence, we assume that Alice and Bob have access to local sources of randomness $\cR_A$ and $\cR_B$, respectively. In order to generate the messages to be transmitted over the public channel they make use of those such that it holds that
\begin{align}
	\Phi_1 = \Phi_1(\cR_A,x^n), \quad \Psi_1 = \Psi_1(\cR_B,y^n)\\
	\Phi_i = \Phi_i(\cR_A,x^n,\Psi^{i-1}), \quad \Psi_i = \Psi_i(\cR_B,y^n,\Phi^{i-1}) \enspace ,
\end{align}
where $\cR_A$ and $\cR_B$ are independent from the jointly distributed random variables $X^n$ and $Y^n$ corresponding the source of common randomness.

We next formalize the secret key rate which is the figure of merit that we would like to maximize for this problem.

\begin{definition} \label{def:secret_key_rate}
	A secret key rate $R$ is called \textbf{achievable} if for every $\varepsilon>0$ and sufficiently large $n$ there exists a secret key agreement scheme such that
	\begin{enumerate}
		\item $Pr(K_A \neq K_B) < \varepsilon$
		\item $I(Z^n,\Phi^\ell,\Psi^\ell;K) < \varepsilon$
		\item $H(K) > R-\varepsilon$
		\item $\log_2(|\cK|) < H(K) + \varepsilon$.
	\end{enumerate}
\end{definition}

We next give some interpretation to the properties that an achievable secret key rate has according to Definition~\ref{def:secret_key_rate}.

The first property basically states that for sufficiently large $n$ the probability that the key at the Alice terminal is unequal to the key at Bob's terminal is arbitrarily small. The second property states that no information can be deduced from the messages shared over the public channel and the source component $Z^n$ about the key $K$. We recall at this point that random variables are stochastically independent if and only if their mutual information is zero and the second property says that we are able to approach this arbitrarily closely.
The third property states that the entropy of $K$ is basically at least $R$ because $\varepsilon$ is arbitrarily small. The fourth property states that the key is almost uniform over the keyspace $\cK$.

The natural question of finding the maximal achievable secret key rate, in the following referred to as the \emph{secret key capacity} as a function of $P_{XYZ}$ has been answered in \cite{ahlswede1993common} and \cite{maurer1993}.

\begin{theorem}\label{th:secret_key_capacity}
	The \textbf{secret key capacity} $\widetilde{C}_S$ denotes the maximal achievable secret key rate for a source of common randomness $P_{XYZ}$ and is bounded by
	\begin{equation} \label{eq:secret_key_cap}
		I(X;Y) - I(X;Z) \leq \widetilde{C}_S \leq I(X;Y|Z)  \enspace .
	\end{equation}
	Furthermore, the secret key capacity is achievable even if we only allow a single transmission over the public channel from Alice to Bob or from Bob to Alice.
\end{theorem}
\begin{corollary}[\cite{bloch2011physical}, Corollary 4.1]
	If $X \fooA Y \fooA Z$ it holds that
	\begin{equation}
		\widetilde{C}_S = I(X;Y) - I(X;Z) = I(X;Y|Z)
	\end{equation}
	and hence upper and lower bound in \eqref{eq:secret_key_cap} are matching.
\end{corollary}
\begin{remark}
	Notice that the secret key capacity $C_S$ is very similar to the secrecy capacity of a degraded wiretap channel $P_{XYZ}(x,y,z) = P_X(x) P_{Y|X}(y|x) P_{Z|Y}(z|y)$ (see equation~\eqref{eq:degraded_wiretap_channel}) except for the fact that a maximization over the distribution $P_X$ is omitted.
\end{remark}

Before we give a proof sketch for the achievability part of Theorem~\ref{th:secret_key_capacity} we examine some special cases for $P_{XYZ}$. Let us assume that either $Z^n=X^n$ or $Z^n=Y^n$ holds. In this case the secret key capacity is zero. This is intuitively appealing as Eve is able to observe all communication over the public channel and one legitimate party has the same information from the source of common randomness as Eve. On the other hand if $P_{XYZ}(x,y,z) = P_{XY}(x,y) P_Z(z)$, observing $Z^n$ gives no information about the shares $X^n$ or $Y^n$. Hence, $Z^n$ provides no useful information to Eve at all which is reflected by the fact that $I(X;Z)=0$ in that case.

\begin{remark}
	In \cite{ahlswede1993common} it has first been discussed how to perform secret sharing with common randomness if the eavesdropper only has access to the messages transmitted over the public channel, i.e. the source was of the form $P_{XY}$ and only later to examine the more general case (including $Z$). This is essentially equivalent to the case that $Z$ is independent of $X,Y$. In this work, we chose the approach of directly introducing the model in the more general setting (including $Z$) as was also done in \cite{maurer1993}.
\end{remark}

\begin{proof}[Proof sketch]
	In the following we sketch how the secret key rate given in Theorem~\ref{th:secret_key_capacity} can be achieved using secrecy codes designed for secure data transmission over a (potentially degraded) wiretap channel. Furthermore, this construction only requires the transmission of a single block of length $n$ from Alice to Bob over the public channel. We follow the structure of the proof given in \cite[Chapter 4.2.1]{bloch2011physical}.
	
	Let the secret key be encoded into a block of $n$ symbols over $\cX$, labelled by $u^n$.
	Alice would like to securely communicate the $i$-th symbol $u_i \in \cX$ to Bob. The choice of this symbol is stochastically independent of the common randomness outputs $X^n,Y^n,Z^n$. She computes $u_i + x_i$ and sends the result over the public channel, where the addition is taken $\mod |\cX|$. Bob receives his $i$-th dedicated symbol from the source of common randomness $y_i$ and $u_i + x_i$ from the public channel. Eve receives $z_i$ from the common randomness source and is able to eavesdrop $u_i + x_i$. This scenario can be interpreted as a wiretap channel. The channel's input is $U$, while the legitimate user's channel output is formed by the tuple $(U+X,Y)$ and the eavesdropper's output is formed by $(U+X,Z)$.
	
	Assume that the sequence all possible sequences $u^n$ are codewords of a wiretap code. It is possible to construct such a code by sampling the elements of all codewords from a single distribution $P_U$ that can be chosen arbitrarily. According to the standard achievability proofs for secrecy codes of wiretap channels (see for instance \cite[Chapter 3.4.1]{bloch2011physical}) 
	From Theorem~\ref{th:cap_wtc} we know that a secrecy code for this channel with rate
	\begin{equation}
		R^* \defeq I(U;Y,U + X) - I(U;Z,U + X) = H(U|Z,U + X) - H(U|Y,U + X)
	\end{equation}
	is achievable by choosing the auxiliary random variable $V=U$. In case $X \fooA Y \fooA Z$ it follows that $U \fooA (U+X,Y) \fooA (U+X,Z)$ and by Theorem~\ref{th:cap_dwtc} that $R^*$ is achievable.

	For the choice of $U$ being uniformly distributed over $\cX$ and using properties of the one-time pad one is able to show that
	\begin{equation}
		R^* = I(X;Y) - I(X;Z) \enspace .
	\end{equation}
	Hence, we have constructed a secret sharing scheme using common randomness from the source $P_{XYZ}$ achieving the rate $I(X;Y)-I(X;Z)$.
\end{proof}
The reason we provide this proof is that the methodology to perform secret sharing using common randomness using secrecy codes for wiretap channels naturally carries over to the finite blocklength regime. We come back to this point in Section~\ref{sec:secret_sharing_oneway}.

\subsection{Helper Data Algorithm} \label{subsec:HDA}

\begin{figure}[t!]
	\begin{centering}
		\resizebox{\textwidth}{!}{
			\includegraphics{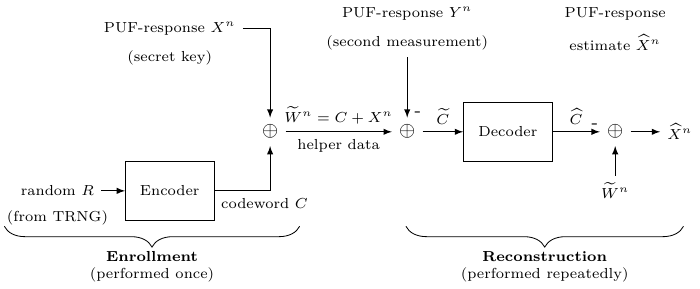}
		}
		\caption{Simplified schematic of a key generation scheme based on a PUF.}
		\label{fig:keyGenOverview}
	\end{centering}
\end{figure}

In the following we present a widely used method to integrate error correction capabilities into the generation of a cryptographic key using PUFs. We refer to this methodology as the helper data algorithm (HDA) \cite{Dodis2004}. A block diagram for the helper data algorithm is shown in Fig.~\ref{fig:keyGenOverview}.

The PUF response $X^n$, in this example the content of multiple SRAM cells after power up, is measured during the manufacturing process of the device, which is referred to as \emph{enrollment}. Furthermore, during the enrollment a random number $R$ is sampled from a True Random Number Generator (TRNG) which is used to select a codeword $C$ at random from an ECC and the helper data $\widetilde{W}^n$ is computed according to $\widetilde{W}^n \defeq C + X^n$. This helper data is published, e.g., in an external storage on the embedded system. Notice that the amount of SRAM cells corresponds to the length of the code in this case. We remark that additions and subtractions in this section are usually performed within the finite field over which the ECC is defined.

The goal of the helper data algorithm is to obtain the PUF response measured during the enrollment at another time when the PUF is measured again. We call this process \emph{reconstruction}. The PUF measurement during the reconstruction phase is denoted by the random variable $Y^n$. Using the helper data $\widetilde{W}^n$ we compute $\widetilde{C}\defeq \widetilde{W}^n-Y^n$. Since $Y^n$ does not have to be equal $X^n$ even if the same device is used during the reconstruction phase we have that $\widetilde{C}=C+E^n$, where $E^n=X^n-Y^n$ denotes the error vector. If $E^n$ is of sufficiently small weight, the codeword is correctly decoded. We denote the decoder's output by $\widehat{C}$. We then estimate the PUF response after applying error correction $\widehat{X}^n$ by computing $\widehat{X}^n = \widetilde{W}^n-\widehat{C}$.

The purpose of the helper data algorithm can be abstracted to the fact that it is not possible to perform coding over the channel from $X^n$ to $Y^n$ specified by the probability mass function $P_{Y^n|X^n}(y^n|x^n)$ as those values are the outputs of the PUF measurement without inherent structure. This channel is a function of the measurement noise and potential temperature dependence or aging effects. The helper data algorithm enables the integration of a structured ECC that can be chosen by the designer. This enables error correction within the reconstruction phase. Notice that the elements of the codeword $C$ and the first estimate before decoding $\widetilde{C}$ are connected by virtually the same channel is $C$ and $X^n$ as well as $\widetilde{C}$ and $Y^n$ differ only additively by the helper data $\widetilde{W}^n$. From the construction it also becomes obvious that the security level of the helper data scheme is upper bounded by the code dimension of the ECC in bits as this is simply the brute force complexity. The length of the code on the other hand determines the required length of the PUF-response and hence can be directly associated with the hardware complexity of the PUF, i.e. via the required number of SRAM cells for an SRAM PUF.

To remove any bias from the PUF-response, frequently the generated secret key is not equal to the PUF-response $X^n$ but corresponds to the hash-value of $X^n$ using a cryptographically secure hash-function. In general an approximate reconstruction of $X^n$ is insufficient for key derivation because slight changes at the input of cryptographic functions, e.g. via the key, typically already lead to substantial changes at the function's output.

\subsection{Helper Data Algorithm for Analog PUFs} \label{subsec:hda_analog}

In contrast to a PUF outputting digital values like in the SRAM PUF, the output of analog PUFs cannot be directly fed into the helper data algorithm introduced in Section~\ref{subsec:HDA}. Therefore, during the enrollment phase the analog PUF output $X^n$ needs to be quantized by an input quantizer and $S^n=Q(X^n)$. Furthermore, additional helper data is generated from $X^n$ for the quantization required in the reconstruction phase. We denote this function by $g$ and the resulting quantization helper data by $W^n$, in particular $W^n=g(X^n)$. A block diagram illustrating the additional quantization and helper data generation steps is given in Fig.~\ref{fig:HDA_analog}. Input quantization and helper data generation are performed elementwise. With slight inaccuracy in notation, we sometimes also write $Q(X)$ for the quantization of a single PUF value. $g(X)$ is treated the same way. We denote the codomain of $Q$ by $\cS^n$ and the codomain of $g$ by $\cW^n$. The quantization helper data $W^n$ is published while $S^n$ forms the key or is the preimage of the key via a cryptographically secure hash function and hence is kept secret. Like in the discrete case, the ECC is fixed and a random number $R$ sampled from a TRNG to determine a codeword $C$. Similar to digital PUFs, the helper data algorithm computes and publishes $\widetilde{W}^n = C + S^n$.

\begin{figure}[t!]
	\begin{centering}
		\resizebox{\textwidth}{!}{
			\includegraphics{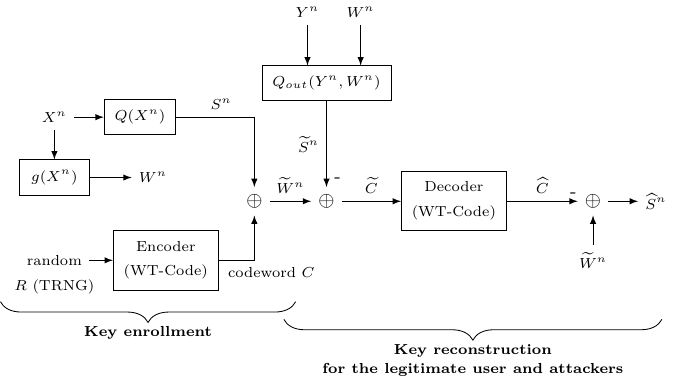}
		}
		\caption{Key enrollment and key reconstruction for analog PUFs}
		\label{fig:HDA_analog}
	\end{centering}
\end{figure}

During the reconstruction phase, the user measures the PUF again. We denote the outcome of this measurement by $Y^n$. From the quantization helper data $W^n$ an output quantizer is derived which is then used to output an estimate $\widetilde{S}^n$ of $S^n$. The remaining steps are analogous to the discrete case. The decoder basically decodes the erroneous codeword $\widetilde{C} = \widetilde{W}^n - \widetilde{S}^n$. Then the decoder outputs $\widehat{C}$ and $\widehat{S}^n = \widetilde{W}^n-\widehat{C}$ is computed. The key will be correctly recovered if $S^n = \widehat{S}^n$.

Since the quantization helper data $W^n$ is derived directly from the PUF measurement during the enrollment phase $X^n$ and so is the secret $S^n$, it has to be assured that from the public $W^n$ it is impossible to derive information about $S^n$. 
The next section therefore deals with establishing quantization helper data that does not leak information about the secret $S^n$.

\subsection{Zero Leakage Helper Data}\label{subsec:zero_leakage_hds}
Next, we give the definition of zero leakage for helper data algorithms. This is the notion that we are aiming at for the generation of the quantization helper data in analog PUFs.

\begin{definition}[\cite{GSV+16}]\label{def:zero_leakage}
	A helper data algorithm is defined to have \emph{zero leakage} if the PUF response (after quantization) $S^n$ and the quantization helper data $W^n$ are stochastically independent, i.e.
	\begin{equation*}
		P_{S^n|W^n}(s^n|w^n) = P_{S^n}(s^n), \; \forall s^n \in \mathcal{S}^n, w^n \in \mathcal{W}^n \enspace .
	\end{equation*}
\end{definition}
As in \cite{GSV+16} we use a slightly stricter definition of zero leakage compared to the more standard definition in \cite{verbitskiy2010key} to avoid pathological cases because we have continuous values for the helper data $W^n$.

Sufficient and necessary conditions for helper data featuring zero information leakage about the secret are given in \cite{GSV+16}. In this work it was shown that it is possible to construct a zero leakage helper data scheme using a function $g$ on the PUF response to generate the helper data scheme, having the following properties:
\begin{enumerate}
	\item $g$ is strictly monotonous function (and therefore an injective) function from each quantization interval to the domain of the helper data $\mathcal{W}$.
	\item Any other function $g^*$ generating the helper data cannot lead to a better reconstruction performance or does not have the zero leakage property.
\end{enumerate}

The main result that we are using is given below:

\begin{theorem}[Thm. 4.8~\cite{GSV+16}]\label{th:helpder_data}
	Let $g$ be monotonously increasing on each quantization interval $A_t$, with $g(A_0) = \ldots = g(A_{N-1})=\mathcal{W}$, where $N$ denotes the number of quantization levels. Let $x_t$ and $x_u$ be from different intervals with $g(x_t)=g(x_u)$. Then in order to satisfy zero leakage the following condition is sufficient and necessary:
	\begin{equation}
		\frac{F_X(x_t)-F_X(q_t)}{p_t} = \frac{F_X(x_u)-F_X(q_u)}{p_u} \enspace ,
	\end{equation}
	where $q_t$ denotes the left border of the interval $A_t$ and $p_t$ denotes the probability that $X$ is sampled to be in $A_t$, i.e. $p_t = Pr(X \in A_t)$. Analogous statements hold for $q_u$, $A_u$ and $p_u$.
\end{theorem}

We refer to points $x_t, x_u$ from different quantization intervals leading to the same helper data as \emph{sibling points}.
Theorem~\ref{th:helpder_data} also leads to a natural way of computing helper data via
\begin{equation}\label{eq:helper_data}
	w = g(x_{t,w}) := \frac{F_X(x_{t,w})-F_X(q_t)}{p_t} \enspace ,
\end{equation}
where $x_{t,w}$ denotes the point within the interval $A_t$ leading to a helper data value of $w$.
This is not the only optimal way to define the helper data but there is no way that leads to better performance during reconstruction while keeping the zero leakage property. Therefore, we take this approach for computing helper data throughout this work.

\begin{lemma}[\cite{stanko2017optimized} Lem. 1]
	The distribution of the helper data defined in equation~\eqref{eq:helper_data} is the uniform distribution over the interval $[0,1]$, i.e.
	\begin{equation}
		f_W(w) = \begin{cases}
			1 \quad \text{for } w \in [0,1]\\
			0 \quad \text{otherwise} \enspace .
		\end{cases}
	\end{equation}
\end{lemma}

Note that this construction facilitates to use the previously discussed equiprobable and equidistant input quantizations into zero leakage helper data algorithms.

In order to obtain the estimate $\widetilde{S}$ from the measurement during reconstruction $Y$ the helper data is used to generate another quantizer. Its interval borders depend on $W$. For an equiprobable quantization, meaning that the input quantizer is formed such that $S$ is uniformly distributed over its range, the construction of the output quantizer is also given in \cite{GSV+16}.

The more general case for arbitrary input distributions has been investigated in \cite{stanko2017optimized}. The authors present an output quantizer for zero leakage helper data that we are also using in the following. The computation of the interval borders of this quantizer is given in Theorem~\ref{th:output_intervals}. We refer to Remark~\ref{rem:tau} for a more specific statement on what is meant by $p_{t-1} \not \ll p_t$.

\begin{theorem}[\cite{stanko2017optimized} Thm. 1]\label{th:output_intervals}
	Let the values $\tau_0,\ldots,\tau_N$ denote the interval borders of the output quantizer used to obtain $\widetilde{S}$ from $Y$ and let $p_{t-1} \not \ll p_t$. Let $\tau_0 = -\infty$ and $\tau_N = \infty$. Let $g_{t}^{-1}(w)$ be the unique value $x$ in quantization interval $t$ such that $g(x)=w$. Then choosing $\tau_t$ iteratively starting from following equation~\eqref{eq:output_intervals} gives the best reconstruction estimate for zero leakage helper data.
	\begin{equation}\label{eq:output_intervals}
		\tau_s = \frac{\ln \left(\frac{p_{t-1}}{p_t}\right)}{g_t^{-1}(w)-g_{t-1}^{-1}(w)} \sigma_N^2 + \frac{g_{t-1}^{-1}(w)+g_t^{-1}(w)}{2}
	\end{equation}
\end{theorem}
\begin{remark}
	Notice that the choice for the output intervals given in Theorem~\ref{th:output_intervals} is consistent with Theorem 5.2 in \cite{GSV+16} for uniform $S$.
\end{remark}
\begin{remark}\label{rem:tau}
	If $p_s \ll p_{s-1}$ holds the symbol $s\in \mathcal{S}$ may be suboptimal irrespective of the channel output. This happens because the a-priori probability $p_{s}$ of the symbol $s$ is so large that for the equality point $\tau_s$ splitting the decision regions between the symbols $s-1$ and $s$ it holds that $\tau_s<\tau_{s-1}$. If this happens the output quantizer has only $|\mathcal{S}|-1$ symbols and we compute
	\begin{equation}
		\tau^* = \frac{\ln \left (\frac{p_{s-2}}{p_s} \right )}{g_s^{-1}(w) - g_{s-1}^{-1}(w)} \sigma_N^2 + \frac{g_{s-1}^{-1}(w)+g_s^{-1}(w)}{2} \enspace .
	\end{equation}
	In case $\tau^*<\tau_{s-2}$ we repeat the procedure. In this case the quantizer has only $|\mathcal{S}|-2$ symbols.
\end{remark}

The computation of the output intervals can either be done during the reconstruction phase or within the enrolment phase. In the latter case the helper data is not a single number $w \in [0,1)$ but rather the set of all interval borders $\tau_s$. This means that one can trade computational complexity during the reconstruction phase against storage consumption within the device, where the helper data needs to be stored.

\begin{remark}\label{rem:zero_leakage}
	Observe that the reconstruction quantizer depends on the distribution of $S$. This distribution in turn depends on the input quantizer and since $\widetilde{S}$ depends on the reconstruction quantizer we have the peculiar case that the conditional pmf $P_{\widetilde{S}|S,W}$ in general depends on the input distribution $P_S$.
\end{remark}

\subsection{Relation between Secret Sharing with Common Randomness and PUFs} \label{subsec:connection_hda_secret_sharing}
Recalling the secret sharing problem using common randomness presented in Section~\ref{subsec:secret_sharing} and comparing it to the problem of reconstructing the value of a PUF during the reconstruction phase that has previously been measured throughout the enrollment phase, we observe that the problems are almost equivalent. This equivalence is outlined in the following.

The PUF responses during enrollment $X^n$ and reconstruction $Y^n$ can be interpreted as $n$ samples from a source of common randomness specified by a joint distribution $P_{XY}$. Notice that in this model there is no adversarial output $Z$ for the moment. The helper data for reconstruction $\widetilde{W}^n$ which is published by the helper data algorithm can be interpreted as data being sent over the public channel by the terminal having the enrollment data $X^n$ (Alice in the secret sharing scenario). The reconstruction measurement data $Y^n$ (Bob's share of the source of common randomness) is then used together with the message $\widetilde{W}^n$ received from the public channel. For a random code with codewords that are iid sampled from the uniform distribution over the range of the RV $X$, the achievability proof sketch of Theorem~\ref{th:secret_key_capacity} shows that the helper data scheme achieves the secret key capacity for the secret sharing problem using the resulting correlated source emanating from the PUF reconstruction problem. Notice that it is not possible to perform forward and backward transmissions over the public channel in this equivalence between secret sharing using common randomness and the PUF reconstruction problem. In the asymptotic setting it has been shown that one-way communication suffices to achieve the secret key capacity. In the finite blocklength regime this is unclear or even suggested not to hold (see Section~\ref{subsec:secret_sharing}). The construction achieving the secret key rate in Theorem~\ref{th:hayashi_secret_key} requires two-way communication over the public channel. Hence, it is unclear whether it is possible to generate a (possibly different) helper data scheme such that the rate in Equation~\eqref{eq:rate_hayashi} can be achieved for the coderate of the ECC. In fact, it is not obvious whether the helper data algorithm is optimal in maximizing the achievable rate if we restrict ourselves to one-way communication over the public channel. This question however is out of scope of this work.

\section{Finite Blocklength Information Theory}\label{sec:finite_blocklength}
In the previous sections we summarized results that make claims in the asymptotic setting as $n$ goes to infinity. However, in practice we have to limit ourselves to some finite value for the blocklength $n$. For data transmission over point to point channels, natural questions to be asked are which rate $R$ can be achieved for a given block length $n$ and fixed block error probability $P_e$ or which $P_e$ can be achieved for fixed $R$ and $n$.

\subsection{Degraded Wiretap and Wiretap Channels}\label{subsec:finite_blocklength_wiretap}
The finite blocklength behaviour of degraded wiretap and wiretap channels has already been investigated in \cite{yang2016finite,yang2019wiretap}. In the finite blocklength regime the achievable rates depend on the concrete values for block error probability of the legitimate user $\varepsilon$ and the security parameter $\delta$. In contrast, similar to DMCs those parameters can be made arbitrarily small in the asymptotic regime as long as the rate is below the secrecy capacity (see Section~\ref{subsec:wiretap_channel}).

The theorems dealing with the asymptotic behaviour (Theorem~\ref{th:cap_dwtc} and Theorem~\ref{th:cap_wtc}) merely provide the information that
there exists a secrecy code of message cardinality
\begin{equation}\label{eq:asymp_cardinality}
	|\mathcal{M}| = 2^{nC_S + o(n)} \enspace .
\end{equation}
$C_S$ is only an approximation of $R^*(n,\varepsilon,\delta)$ and estimating $R^*(n,\varepsilon,\delta)$ by $C_S$ is only reasonable for very large $n$. For small and moderate $n$ it is essential to further analyze the $o(n)$ term in Equation~\eqref{eq:asymp_cardinality}.
The following theorem established in \cite{yang2019wiretap} gives more precise upper and lower bounds on the exponent in Equation~\eqref{eq:asymp_cardinality}. The error term for these bounds is in the order of $\mathcal{O}(\log(n)/n)$.

\begin{theorem}[\cite{yang2019wiretap} Thm. 13]\label{th:finite_wiretap}
	For a discrete memoryless (degraded or general) wiretap channel $P_{BE|A}$ with secrecy capacity $C_S$ and for $\varepsilon+\delta < 1$ it holds that
	\begin{equation}
		R^*_{max}(n,\varepsilon,\delta) \geq C_S  - \sqrt{\frac{V_1}{n}} Q^{-1}(\varepsilon) - \sqrt{\frac{V_2}{n}} Q^{-1}(\delta) + \mathcal{O}\left(\frac{\log(n)}{n}\right)
	\end{equation}
	and
	\begin{equation}
		R^*_{avg}(n,\varepsilon,\delta) \leq C_S - \sqrt{\frac{V_c}{n}} Q^{-1}(\varepsilon+\delta) + \mathcal{O}\left(\frac{\log(n)}{n}\right) \enspace ,
	\end{equation}
	where
	\begin{equation}
		V_1 := \sum_{a\in\mathcal{A}} P_A(a) \left( \sum_{b\in \mathcal{B}} P_{B|A}(b|a) \log_2^2 \left( \frac{P_{B|A}(b|a)}{P_B(b)} \right) - D(P_{B|A=a}||P_B)^2 \right) \enspace ,
	\end{equation}
	\begin{equation}
		V_2 := \sum_{a \in \mathcal{A}} P_A(a) \left( \sum_{e\in \mathcal{E}} P_{E|A}(e|a) \log_2^2 \left( \frac{P_{E|A}(e|a)}{P_E(e)} \right) - D(P_{E|A=a}||P_E)^2 \right) \enspace ,
	\end{equation}
	and
	\begin{align}
		V_c := \sum_{a \in \mathcal{A}} P_A(a) \Bigg( &\sum_{b\in\mathcal{B},e\in\mathcal{E}} P_{BE|A}(b,e|a) \log^2\left(\frac{P_{BE|A}(b,e|a)}{P_{E|A}(e|a)P_{B|E}(b|e)}\right)\\ \nonumber 
		&- D(P_{BE|A=a}||P_{B|E}P_{E|A=a})^2 \Bigg) \enspace .
	\end{align}
	$Q^{-1}$ denotes the inverse of $Q(x) = \int_x^\infty \tfrac{1}{\sqrt{2\pi}} \exp\left(-\tfrac{z^2}{2}\right) \diff z$.
\end{theorem}
To this point, it is unclear to us whether the construction using wiretap coding is optimal in terms of maximizing the secret key rate in the finite blocklength regime. To the best of our knowledge it is even unknown whether it is possible to achieve the secret key rate in Theorem~\ref{th:hayashi_secret_key} using one-way communication over the public channel or not. For the interested reader, we refer to the discussion on the necessity of two-way communication to achieve the secret key rate determined by Theorem~\ref{th:hayashi_secret_key} in \cite[Section VII]{hayashi2016}.

\subsection{Secret Sharing using Common Randomness}\label{subsec:finite_secret_sharing}
In this section we provide results for secret sharing with common randomness in the finite blocklength regime. This problem has been studied in \cite{hayashi2016}. The authors established the secret key rate up to an error term in the order of $\mathcal{O}(\log(n)/n)$.

For the finite length setting we use the secret key definition given in \cite{hayashi2016} which we present below. Notice that the security condition matches with the one of secrecy codes for wiretap channels given in Definition~\ref{def:secrecy_coding_strategy}.

\begin{definition}
    A secret sharing protocol is defined to achieve $\varepsilon$ reliability and with average secrecy parameter $\delta$ if the probability that the two legitimate partners Alice and Bob fail to agree on the same key with probability less than $\varepsilon$ and it holds that
    \begin{equation}\label{eq:avg_security_secret_sharing}
	d(P_{K,\Phi^{\ell},\Psi^{\ell},Z^n},P_K^{unif} P_{\Phi^{\ell},\Psi^{\ell},Z^n}) \leq \delta \enspace .
    \end{equation}
    where $K$ denotes the secret key, $Z^n$ denotes the share of the source of common randomness that the eavesdropper obtains and $\Phi^{\ell},\Psi^{\ell}$ denote the communication over the public channel.
    For maximum secrecy the security condition is changed to
    \begin{equation}
        \max_{k\in \cK} d(P_{Z^n|K=k},Q_{Z^n,\Phi^{\ell},\Psi^{\ell}}) \leq \delta \enspace ,
    \end{equation}
    where $Q_{Z^n,\Phi^{\ell},\Psi^{\ell}}$ denotes the marginal distribution of $(Z^n,\Phi^{\ell},\Psi^{\ell})$ if uniformly distributed keys are considered.
    We define the maximal achievable secret key rate with blocklength $n$ by
    \begin{equation}
	\widetilde{R}^*(n,\varepsilon,\delta):= \max\left\{\frac{\log(|\mathcal{K}|)}{n}: \exists (n,|\mathcal{K}|,\varepsilon,\delta) \text{ secret sharing protocol}\right\}
    \end{equation}
    and specify by indices whether we mean the average or the maximum secrecy definition.
\end{definition}

\begin{theorem}[\cite{hayashi2016},Thm.15] \label{th:hayashi_secret_key}
	For every $\varepsilon,\delta>0$ such that $\varepsilon+\delta<1$ and iid $(X^n,Y^n,Z^n)$ sampled according to joint pmf $P_{XYZ}$ such that $X \fooA Y \fooA Z$, the maximal secret key rate for $n,\varepsilon,\delta$ is given by
	\begin{equation}\label{eq:rate_hayashi}
		\widetilde{R}_{avg}^*(n,\varepsilon,\delta) = C_S - \sqrt{\frac{V_c'}{n}} Q^{-1} (\varepsilon+\delta) + \mathcal{O}\left(\frac{\log(n)}{n}\right) \enspace ,
	\end{equation}
	where
	\begin{align}
		V_c' := \sum_{x \in \mathcal{X}} P_X(x) \Bigg( &\sum_{y\in\mathcal{Y},z\in\mathcal{Z}} P_{YZ|X}(y,z|x) \log_2^2\left(\frac{P_{YZ|X}(y,z|x)}{P_{Z|X}(z|x)P_{Y|Z}(y|z)}\right)\\ \nonumber 
		&- D(P_{YZ|X=x}||P_{Y|Z}P_{Z|X=x})^2 \Bigg) \enspace .
	\end{align}
\end{theorem}

Examining equation~\eqref{eq:rate_hayashi}, we observe that its structure is similar to the bounds on the maximal secrecy rate in Theorem~\ref{th:finite_wiretap}. The first term is the asymptotic result, i.e. the secret key capacity defined by the source of common randomness. This value is independent of $n,\varepsilon,\delta$ and only depends on $P_{XYZ}$. Another term that depends on $P_{XYZ}$ via the dispersion coefficient $V_c$ but also on $n,\varepsilon,\delta$ is subtracted from this value. Notably, this term decreases in $n$ with speed $1/\sqrt{n}$ and hence Theorem~\ref{th:hayashi_secret_key} is consistent with the asymptotic result presented in Theorem~\ref{th:secret_key_capacity}. Finally, there is an error term in the order of $\mathcal{O}(\log(n)/n)$ that decreases significantly faster than $1/\sqrt{n}$ such that ignoring it approximates reality reasonably well for moderate values of $n$.

\begin{remark}\label{rem:secret_sharing}
The secret sharing protocol achieving the maximal secrecy rate $\widetilde{R}^*_{avg}(n,\varepsilon,\delta)$ presented in \cite{hayashi2016} requires two-way communication over the public channel. This is in contrast to the achievability proof in the asymptotic setting (Theorem~\ref{th:secret_key_capacity}) for which one-way communication suffices.
\end{remark}

This limitation of requiring two-way communication in achievability proof has practical consequences, in particular in the context of applying the result to Physical Unclonable Functions as for this purpose a protocol only requiring one-way communication is necessary.

\section{PUF-Based Tamper Protection Foil}\label{sec:channel_model}

Going from the pure PUF-based key generation scheme to tamper protection brings additional requirements for the PUF. The combination of enrollment phase and reconstruction phase using HDAs as proposed in Section~\ref{subsec:hda_analog} for the combination of legitimate user and a physical attacker is related to secure communication over a discrete memoryless degraded wiretap channel or a discrete memoryless wiretap channel depending on the chosen attacker models that are specified in the following.

\subsection{Attacker Model}\label{subsec:attacker_model}

Aside from the standard attack vectors of key generation schemes with analog PUFs, like helper data leakage, in addition the physical tampering aims at obtaining secret information, which needs to be prevented by the designer.

Considering all attack vectors, the remaining security of the system needs to exceed a specified minimum security level such as 120 bit, as currently recommended by the BSI \cite{keylength}. In this work, we chose to investigate 128, 192 and 256 Bit security levels as they are common in many cryptographic applications.

\paragraph{Leakage through Helper Data}

Since early approaches like the fuzzy commitment or fuzzy extractor \cite{DRS04}, helper data leakage is a topic to be considered when designing secure key generation with noisy secrets, and in particular for PUFs \cite{DGSV15} when imperfections come into play \cite{DGV+16}. For the remainder of the this work, we assume a random number with full entropy for selecting the codeword and refer to the mentioned related works to design helper data schemes that avoid leakage through $\widetilde{W}^n$ by the helper data algorithm.
In addition to $\widetilde{W}^n$, the quantization helper data $W^n$ needs to be considered for analog PUFs. We therefore apply the zero leakage helper data generation proposed in Section~\ref{subsec:zero_leakage_hds} for $W^n$ throughout this work.

\paragraph{Physical Tampering} 

As physical tamper protection aims to withstand attackers being able to use sophisticated tools, a wide range of attacks needs to be considered \cite{weingart2000physical,Imm19,GSHO21,jilHDSB}, where a special emphasis is put to physical drilling, that affects only very small areas.

As reference, for the system proposed in \cite{IOK+18}, drilling with a conventional drill bit with diameter $>300 \mu m$ fully destroys one electrode in both layers, such that 23 out of 128, or $18\%$ of the capacitance values are destroyed. Attacks during operation can be detected with the integrated run-time tamper detection measures to bring the system in a secure state. In contrast, attacks on the powered-off device are more challenging from a theoretical point of view. For identical layouts, we need to assume that the attacker knows the position of the destroyed PUF cells within the PUF response, which we model as erasures in the proposed coding schemes.

A neuralgic point for an attack lies between the measurement circuit and the key generation, as the attacker can obtain the digitized low-noise values from tapping a single wire, before the values are interpreted in the embedded key management system and the attack detection triggers. The attacker can only perform a single measurement of the PUF in this case because after this measurement the system will notice that the foil has been attacked and goes into a secure state, eliminating the possibility for further measurements. In the following, this attacker is referred to as the ``digital attacker`` for brevity.

In another scenario, the attacker performs more advanced analog measurements, i.e. we assume he affords better measurement equipment. We call this kind of attacker the ``analog attacker`` for short. We make the very conservative assumption the attacker can measure with infinite precision, i.e. his equipment has infinitely fine quantization steps and furthermore he can perform an unlimited number of measurements. Due to the unlimited amount of measurements with uncorrelated measurement noise, the attacker is able to apply post-processing to effectively eliminate the measurement noise. However, in this scenario the attacker needs to drill multiple holes to tap multiple analog wires, potentially with a larger wire diameter instead of only one small hole to probe a digital signal. We will show in Section~\ref{sec:securing_foil} that this attack is indeed problematic, especially for larger field sizes and hence propose to apply countermeasures on the hardware level eliminating the possibility for an attacker to perform those advanced measurements. A profound reasoning for the amount of PUF cells being destroyed by this attack is out of scope of this work. It is reasonable though to assume that a significant extra proportion of the foil needs to be destroyed compared to the digital attacker to mount such an attack. In our examples, we used twice the amount of destroyed capacitances compared to the digital attacker.

Within the helper data algorithm the ECC enables recovering the key from an erroneous PUF measurement as long as the number of errors is sufficiently small. Designers therefore need to take care that an attacker cannot retrieve information about the secret key using the redundancy inflicted by the ECC even though the attacker has to destroy a fraction of the PUF cells to measure the PUF. Investigating fundamental limits of this problem is the major contribution of this work.

In addition to the possibility of wiretapping internal signals, the cable connections through the foil that could reveal timing or global power side channels need to be taken into account. This topic will not be considered in the scope of this work and needs to be addressed individually when designing the host system.

\subsection{Enrollment and Reconstruction Phase for the Foil PUF, Legitimate User} \label{subsec:channel_legitimate}
The PUF considered in this work is established by a foil consisting of electrodes forming a mesh of capacitances. The purpose of this PUF is twofold. First it is used for storing a cryptographic key in a secure manner and secondly it shall protect the components within its inside. An attacker trying to perform a side channel attack is required to measure parameters related to the implementation and the foil prevents him from accessing critical hardware. As composition of physical variations, the capacitances can be modelled as Gaussian random variables \cite{IOK+18}. While the work in \cite{GXKF22} focuses on the reliability side with a strong focus on the specific implication and its measurements, this work aims at quantifying the information theoretic security limits of this type of PUF.

The distribution of the differentially evaluated PUF response for a single cell follows a Gaussian of zero mean and variance $\sigma_P^2$, i.e. $X \sim \mathcal{N}(0,\sigma_P^2)$. We assume that the PUF values for the individual cells to be independently identically distributed (iid). We denote the output one PUF cell during the enrollment phase by the random variable $X$. To obtain the secret $S$ for this cell, the measured data is quantized (see Section~\ref{subsec:hda_analog}). The generation of the helper data $W$of a PUF is peculiar because not only shall it be of value within the reconstruction process but have the zero leakage property (Definition~\ref{def:zero_leakage}). Zero-leakage quantization helper data generation has already been discussed in Section~\ref{subsec:zero_leakage_hds} and we follow this approach throughout this work. In Fig.~\ref{fig:enrolment} the pdf of a single PUF cell is depicted. The red dot represents the realization of $X$. In this example we used 2 bit quantization to obtain the value of the secret value of the PUF, denoted by $S$. The dashed lines represent the borders of the quantizer and the numbers inside the intervals represent the respective values of $S$, i.e. in our example $S=3$. The location of $X$ within the respective interval is then used to derive the helper data $W$ which is later on used to find an estimate for $S$ during the reconstruction phase. 

\begin{figure}[t]
	\begin{center}
		\includegraphics{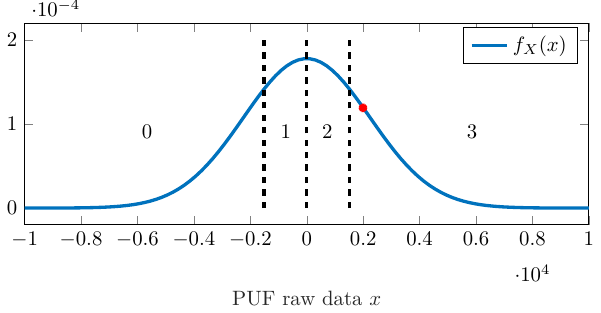}
	\end{center}
	\caption{Enrollment phase}
	\label{fig:enrolment}
\end{figure}
\begin{figure}[t]
	\begin{center}
		\includegraphics{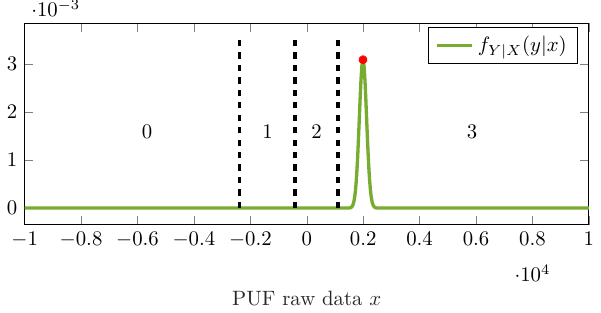}
	\end{center}
	\caption{Reconstruction phase}
	\label{fig:reconstruction}
\end{figure}

The PUF response during the reconstruction phase is modelled as the PUF response during the enrolment phase $X$ perturbed by additive Gaussian noise with variance $\sigma_N^2$. We denote this output by the random variable $Y=X+N$ with $N \sim \mathcal{N}(0,\sigma_N^2)$. During the reconstruction phase, $Y$ is combined with the helper data $W$ to output an estimate for $S$, in the following denoted as $\widetilde{S}$. Throughout this work we use $\sigma_P = 2241$ and $\sigma_N = 129$ which is consistent with the results in \cite{GXKF22}. The output quantizer $Q_{out}$ depends on the quantization helper data $W$ and has been specified in Section~\ref{subsec:zero_leakage_hds}. We denote the output of this quantizer by $\widetilde{S}$.

The conditional distribution of the PUF response during the reconstruction phase for the example in Fig.~\ref{fig:enrolment} is shown in Fig.~\ref{fig:reconstruction}. Notice that the quantization intervals during the reconstruction phase have been shifted compared to the enrollment phase. This is due to the utilization of the helper data $W$. By utilizing the helper data the distance between the value during the enrollment phase (the red dot) and the relevant quantization boundary for the error has increased, leading to a lower reconstruction failure probability.

Our next goal is to secure the HDA for the foil PUF against the attack scenarios mentioned in Section~\ref{subsec:attacker_model}. 
Section~\ref{subsec:connection_hda_secret_sharing} discusses the connection between secret sharing with common randomness and HDAs for PUFs. However, it has also been pointed out that the secret sharing protocol used to establish Theorem~\ref{th:hayashi_secret_key} requires two-way communication, thereby making it unapplicable in the PUF setting. Hence, in the following we investigate secret sharing using common randomness restricted to one-way communication in the finite blocklength regime.

\section{Secret Key Capacities for the Foil PUF}\label{sec:securing_foil}

\begin{figure}[t!]
	\begin{centering}
		\resizebox{\textwidth}{!}{
			\includegraphics{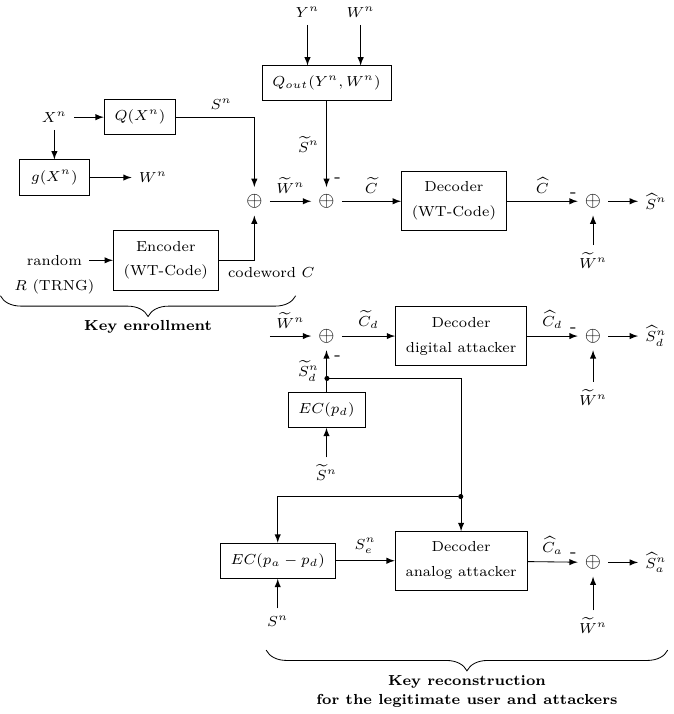}
		}
		\caption{Key enrollment and key reconstruction for legitimate user, digital and analog attacker}
		\label{fig:keyGenOverview_analog}
	\end{centering}
\end{figure}

In this section, we establish the connection between enrollment and reconstruction phase for the legitimate user combined with the attacker models introduced in Subsection~\ref{subsec:attacker_model} and the secret key generation using correlated sources. The enrollment and reconstruction phase during normal operating conditions as well as for the two mentioned attack scenarios are depicted in Fig.~\ref{fig:keyGenOverview_analog}.

We first tackle the problem of determining the maximal code rate of the ECC for a HDA providing security against the attacks described in Section~\ref{subsec:attacker_model}. Essentially, this code rate is related to the hardware complexity and the security level of the PUF as outlined in Section~\ref{subsec:HDA}.

Updating Fig.~\ref{fig:HDA_analog} to also include the attacker models discussed in Section~\ref{subsec:attacker_model} lead to the block diagram shown in Fig.~\ref{fig:keyGenOverview_analog}. Depending on whether we exclude the analog attacker by preventing the necessary measurements on a hardware level as proposed in Section~\ref{subsec:attacker_model} only the variables connected to the decoder for the digital attacker are of concern or also the parts associated with the analog attacker. To keep the analysis of the model simple, we take the approach of examining the scenario where only the digital attacker needs to be considered first. The changes necessary to also integrate the analog measurements of the PUF then build from there.

We observe that for the reconstruction phase the legitimate user measures $Y^n$ and utilizing the quantization helper data $W^n$ obtains an estimate for $S^n$ referred to as $\widetilde{S}^n$. In contrast via the probing the digital measurement line the attacker obtains the same measurement with additional erasures $\widetilde{S}_d^n$. By a similar argument to the one in Section~\ref{subsec:connection_hda_secret_sharing} we have that the resulting helper data scheme's task is to implement a secret sharing algorithm using a source of common randomness. The source of common randomness is formed by the joint probability distribution $P_{S,\widetilde{S},\widetilde{S}_d}$. In the secret sharing problem, the first terminal Alice gets access to $S^n$ and wants to agree on a key with Bob who gets $\widetilde{S}^n$, while the eavesdropper Eve obtains $\widetilde{S}_d^n$. As described in Section~\ref{subsec:secret_sharing}, Alice and Bob next exchange messages over a public channel and aim to agree on a secret key in a secure manner, i.e. in a way such that it is impossible for Eve to obtain information about the key. Interpreting the HDA as a way to integrate a secret sharing protocol, the information shared over the public channel are the quantization helper data $W^n$ and the reconstruction helper data $\widetilde{W}^n$. Alice's and Bob's goal is to agree on a key which is as large as possible, thereby maximizing the entropy of the key for some arbitrary but fixed $n$. Results on the secret key capacity (asymptotic setting as $n\to \infty$) have been recapitulated in Section~\ref{subsec:secret_sharing} (Theorem~\ref{th:secret_key_capacity}) and applying these results to the HDA scheme in Fig.~\ref{fig:keyGenOverview_analog} leads to the result in Theorem~\ref{th:asymptotic_digital}.

\begin{definition}\label{def:key_capacity}
	We consider an HDA with $q$ input quantization levels and an attacker trying to obtain the secret established during the enrollment process.
	The key capacity $C^q_{key,attacker}$ in this setup is the maximal asymptotic code rate for which a wiretap code exists such that the block error probability for the legitimate user is arbitrarily small and the system achieves an arbitrarily high secrecy level as the blocklength $n$ goes to infinity.
\end{definition}
\begin{theorem}\label{th:asymptotic_digital}
	The key capacity $C^q_{key,dig}$ for the foil PUF and the digital attacker is  optimal in the sense that it enables a code rate for the ECC that is equal to the secret key capacity of the secret sharing problem with common randomness specified by the source $P_{S,\widetilde{S},\widetilde{S}_d}$. More specifically, the code rate $C_{key,dig}^q$ is defined by
	\begin{equation}\label{eq:puf_capacity}
		C^q_{key,dig} \defeq \max_{\text{input quantizer}} I(S;\widetilde{S}|W) p_d \enspace .
	\end{equation}
\end{theorem}
\begin{proof}
	First of all it is obvious that it is impossible to achieve a rate higher than the secret key capacity of the secret sharing problem using common randomness specified by the joint distribution $P_{S,\widetilde{S},\widetilde{S}_d}$.
	
	For the achievability consider the codewords of the ECC to be sampled randomly and iid from the uniform distribution over the range of $S$, denoted by $\cS$. The HDA  in Fig.~\ref{fig:keyGenOverview_analog} precisely implements the construction outlined in the achievability proof sketch of Theorem~\ref{th:secret_key_capacity}. Hence, due to the fact that $\widetilde{S}_e$ is defined to be the output of an erasure channel having erasure probability $p_d$.
	
	By the proof sketch for the achievability of Theorem~\ref{th:secret_key_capacity} and by identifying $S\equiv X$, $Y \equiv \widetilde{S}$ and $Z \equiv \widetilde{S}_d$, we have according to Theorem~\ref{th:secret_key_capacity}
	\begin{align}\label{eq:achievability_hda_adversary}
		\widetilde{C}_S &= I(X;Y) - I(X;Z) = I(S;\widetilde{S}|W) - I(S;\widetilde{S}_e|W) \nonumber \\
		&= I(S;\widetilde{S}|W) - I(S;\widetilde{S}|W)(1-p_d) = I(S;\widetilde{S}|W) p_d \enspace .
	\end{align}
	Furthermore, compared to the secret sharing using common randomness problem introduced in Section~\ref{subsec:secret_sharing} we have the additional freedom of choosing the input quantizer. This extra degree of freedom can be exploited since Equation~\eqref{eq:achievability_hda_adversary} is valid for any input quantizer Equation~\eqref{eq:puf_capacity} follows, completing the proof.
\end{proof}

In order to compute this capacity the first step is to analyze the relation between $S$ and $\widetilde{S}$ conditioned on the knowledge of the helper data $W$. The distribution of $S$ follows from the distribution of the PUF during the enrollment phase, i.e. from the distribution of $X$, and from the choice of the input quantizer. It is complicated to perform the maximization in Theorem~\ref{th:asymptotic_digital} because the channel matrix is changing according to the input quantizer as its choice influences the output quantizer (see Remark~\ref{rem:zero_leakage}). The following observation is helpful though.

\begin{lemma}\label{lem:simplification}
	For the channel resulting in the concatenation of enrolment and reconstruction phase at the legitimate user it holds that
	\begin{equation}
		I(S;\widetilde{S}|W) = I(X;\widetilde{S}|W)
	\end{equation}
\end{lemma}
\begin{proof}
	By using the chain rule of mutual information it holds that
	\begin{align}
		I(X,S;\widetilde{S}|W) &= I(S;\widetilde{S}|W) + I(X;\widetilde{S}|W,S)\\
		&= I(X;\widetilde{S}|W) + I(S;\widetilde{S}|W,X) \enspace .
	\end{align}
	It holds that $I(X;\widetilde{S}|W,S) = 0$ and $I(S;\widetilde{S}|W,X) = 0$ because $(W,S)$ uniquely determines $X$ and $X$ determines $S$ and the statement of the Lemma follows.
\end{proof}

Since Lemma~\ref{lem:simplification} shows that we can focus on $I(X;\widetilde{S}|W)$ and furthermore
\begin{align}
	I(X;\widetilde{S}|W) &= H(\widetilde{S}|W) - H(\widetilde{S}|W,X)\\
	&= H(\widetilde{S}|W) - H(\widetilde{S}|X) \enspace,
\end{align}
where the last equality follows because $W$ is a function of $X$.
We observe that the quantity that we are interested in is the difference between the uncertainty of $\widetilde{S}$ given the publicly available helper data $W$ compared to the uncertainty of $\widetilde{S}$ given the private PUF response during the enrollment phase $X$.

Next we deal with the additional analog attacker. This attacker not only gets access to the digitized measurement of the internal circuitry of the PUF after the erasure channel $EC(p_d)$ but also to an analog measurement of the foil PUF passed through another erasure channel. This channel inflicts an additional fraction of $(p_a-p_d)$ errors into the positions that have not yet been erased by the hole necessary for the digital measurement. Hence, this erasure channel needs to access which cells have been erased by the digital attacker and this is reflected in the additional input $\widetilde{S}_d^n$. Its only purpose is to declare which cells have already been erased in this context.

The attacker's decoder however can utilize the information that $\widetilde{S}_d^n$ provides. This information is only useful for reconstructing the values that have been erased by the additional erasures because the number of measurements for the analog attacker is unlimited and hence unerased symbols are perfect reconstructions of the enrollment value $S^n$.

\begin{theorem}\label{th:asymptotic_analog}
	For the key capacity $C_{key,ana}^q$ of the foil PUF and the analog attacker it holds that
	\begin{equation}
		\max_{\text{input quantizer}} I(S;\widetilde{S}|W) (1-p_a+p_d) - H(S)(1-p_a) \leq C^q_{key,ana} \leq \max_{\text{input quantizer}} I(S;\widetilde{S}|W) p_d \enspace .
	\end{equation}
\end{theorem}
\begin{proof}
	We first prove the lower bound. Using Theorem~\ref{th:secret_key_capacity} we have that
	\begin{equation}\label{eq:analog_attacker}
		C^q_{key}\geq I(S;\widetilde{S}|W) - I(S;\widetilde{S}_d,\widetilde{S}_a|W)\enspace .
	\end{equation}
	By the definition of mutual information it holds that
	\begin{equation}
		I(S;\widetilde{S}_d,\widetilde{S}_a|W) = H(S) - H(S|\widetilde{S}_d,\widetilde{S}_a,W)
	\end{equation}
	because $H(S|W)=H(S)$ due to the zero leakage condition and furthermore we observe that
	\begin{align}
		H(S|\widetilde{S}_d,\widetilde{S}_a,W) &=  H(S|\widetilde{S}_d=0,\widetilde{S}_a=E,W) P_{\widetilde{S}_d}(0)(p_a-p_d) \nonumber \\
			&+ \ldots \nonumber \\
			&+ H(S|\widetilde{S}_d=|\cS|,\widetilde{S}_a=E,W) P_{\widetilde{S}_d}(|\cS|)(p_a-p_d) \nonumber \\
			&+ H(S|\widetilde{S}_d=E,\widetilde{S}_a=E,W) \, p_d \nonumber \\
			&= H(S|\widetilde{S},W) (p_a-p_d) + H(S) p_d \enspace .
	\end{align}
	Plugging these simplifications into the right hand side of \eqref{eq:analog_attacker}, we have
	\begin{equation}
		H(S) p_d - H(S|\widetilde{S},W) (1-p_a+p_d)
	\end{equation}
	which is equivalent to the lower bound in the theorem after some simple algebraic steps and maximizing over the choice of the input quantizer.
	
	The upper bound follows trivially because the analog attacker is stronger than the digital attacker. However, we also show that the upper bound cannot be trivially improved by the well known upper bound
	\begin{equation}\label{eq:upper_analog}
		I(S;\widetilde{S}|\widetilde{S}_d,\widetilde{S}_a,W) = H(S|\widetilde{S}_d,\widetilde{S}_a,W) - H(S|\widetilde{S},\widetilde{S}_d,\widetilde{S}_a,W) \enspace .
	\end{equation}
	The first entropy in Equation~\eqref{eq:upper_analog} has already been simplified within this proof. Hence, we only need to simplify $H(S|\widetilde{S},\widetilde{S}_d,\widetilde{S}_a,W)$. Observe that
	\begin{equation}
		H(S|\widetilde{S},\widetilde{S}_d,\widetilde{S}_a,W) = H(S|\widetilde{S},\widetilde{S}_a=E,W)p_a = H(S|\widetilde{S})p_a
	\end{equation}
	and therefore
	\begin{equation}
		I(S;\widetilde{S}|\widetilde{S}_d,\widetilde{S}_a,W) = I(S;\widetilde{S}|W)p_d \enspace .
	\end{equation}
\end{proof}

\begin{table}
	\caption{Achievable asymptotic rates for the PUF-channel for weak and strong attackers, different input quantization alphabets and quantization strategies, erasure probability $p_d=0.1$ for the digital attacker and $p_a=0.2$ for the analog attacker, $\sigma_P = 2241, \sigma_N = 129$}
	\centering
	\begin{tabular}{|c|c|c|c|c|c|c|}
		\hline
		quantizer & equidist. & equidist. & equiprob. & equiprob. & optimized & optimized  \\
		levels & digital & analog & digital & analog & digital & analog\\
		\hline
		2 & 0.1 & 0.1 & 0.1 & 0.1 & 0.1 & 0.1\\
		4 & 0.117 & 0.117 & 0.2 & 0.2 & 0.2 & 0.2\\
		8 & 0.196 & 0.196 & 0.298 & 0.282 & 0.299 & 0.294\\
		16 & 0.291 & 0.291 & 0.356 & 0 & 0.363 & 0.327\\
		32 & 0.376 & 0.269 & 0.382 & 0 & 0.382 & 0.311\\
		64 & 0.389 & 0 & 0.398 & 0 & 0.398 & 0\\
		128 & 0.404 & 0 & 0.406 & 0 & 0.406 & 0\\
		256 & 0.410 & 0 & 0.409 & 0 & 0.410 & 0\\
		\hline
	\end{tabular}
	\label{tab:achievable_rates_0.1_asymptotic}
\end{table}

\begin{table}
	\caption{Achievable asymptotic rates for the PUF-channel for weak and strong attackers, different input quantization alphabets and quantization strategies, erasure probability $p_d=0.18$ for the digital attacker and $p_a = 0.36$ for the analog attacker, $\sigma_P = 2241, \sigma_N = 129$}
	\centering
	\begin{tabular}{|c|c|c|c|c|c|c|}
		\hline
		quantizer & equidist. & equidist. & equiprob. & equiprob. & optimized & optimized  \\
		levels & digital & analog & digital & analog & digital & analog\\
		\hline
		2 & 0.18 & 0.18 & 0.18 & 0.18 & 0.18 & 0.18\\
		4 & 0.211 & 0.211 & 0.36 & 0.36 & 0.36 & 0.36\\
		8 & 0.353 & 0.353 & 0.536 & 0.524 & 0.537 & 0.533\\
		16 & 0.523 & 0.523 & 0.640 & 0.356 & 0.654 & 0.600\\
		32 & 0.677 & 0.591 & 0.687 & 0 & 0.688 & 0.613\\
		64 & 0.700 & 0.061 & 0.716 & 0 & 0.716 & 0.061\\
		128 & 0.727 & 0 & 0.731 & 0 & 0.731 & 0\\
		256 & 0.738 & 0 & 0.737 & 0 & 0.738 & 0\\
		\hline
	\end{tabular}
	\label{tab:achievable_rates_0.18_asymptotic}
\end{table}
Evaluation of Theorem~\ref{th:asymptotic_digital} and Theorem~\ref{th:asymptotic_analog} leads to the achievability results presented in Tables~\ref{tab:achievable_rates_0.1_asymptotic} and \ref{tab:achievable_rates_0.18_asymptotic}. We assumed that the analog attacker has to destroy twice as much of the foil as the digital attacker has to. In Table~\ref{tab:achievable_rates_0.1_asymptotic} we took the conservative approach of assuming that the digital attacker destroys $10\%$ of the PUF cells, i.e. $p_d=0.1, p_a=0.2$. Table~\ref{tab:achievable_rates_0.18_asymptotic} shows the results for $p_d=0.18, p_a=0.36$, which is the value corresponding to destroying one electrode in both layers of the foil as discussed in \cite{IOK+18} considering a hole diameter $>300\mu m$.

In both tables we considered various numbers of quantization levels and input quantizer strategies. From the proofs of Theorem~\ref{th:asymptotic_digital} and Theorem~\ref{th:asymptotic_analog} it is easy to see that choosing a suboptimal input quantizer still gives an achievable lower bound on the secret key rate.

In this work, we consider equidistant and equiprobable input quantization. Furthermore, we use an optimization algorithm aiming to find the input quantizer maximizing the achievable secret key rate according to Theorem~\ref{th:asymptotic_digital} and Theorem~\ref{th:asymptotic_analog}.

The results show that in case we have to cope with the digital attacker that the equiprobable input quantization performs better than equidistant quantization. For the analog attacker the opposite is the case as we increase the number of quantization levels. For a small number of quantization levels equiprobable quantization still performs better.

The reason for that behaviour is that the analog attacker obtains a perfect duplicate of the PUF cell measured during enrollment in case the respective cell is not erased by the holes he is required to drill to perform measurements. Especially if the measurement during the reconstruction by the legitimate user is unreliable, the analog measurements provide valuable additional information to the analog attacker that the legitimate user does not have. The input quantization has a substantial effect on the reliability of the measurement during reconstruction. Once an interval length at the output quantizer (influenced by the input quantizer) is below a certain threshold a measurement result of this value becomes unreliable. In this case the additional information provided by the analog measurement substantially increases due to the insecurity of the reconstruction value for the legitimate user.

This explains why equidistant input quantization performs better for a larger number of quantization levels if the analog attacker is considered. For a small amount of quantization levels the intervals are anyway large enough and the dominant factor is the entropy of the input, which is obviously maximized for the uniform input distribution achieved by the equiprobable input quantizer. For the digital attacker we do not have this trade-off and hence equiprobable quantization always performs better than equidistant quantization.

\section{Secret Sharing using One-Way Communication for Finite Lengths}\label{sec:secret_sharing_oneway}
For this section we stick to the notation used in Section~\ref{subsec:secret_sharing} that is commonly used in secret sharing with common randomness. In particular $X,Y$ and $Z$ are not related to the foil PUF in particular but rather form a general source of randomness via the pmf $P_{XYZ}$. After this section $X$ shall again be considered the analog measurement during enrollment and $Y$ shall again be considered the analog measurement during reconstruction.

Because of this inherent limitation of the construction used in \cite{hayashi2016}, we take the approach of utilizing the construction in the achievability proof sketch of Theorem~\ref{th:secret_key_capacity} also in the finite blocklength case.

\begin{theorem}\label{th:finite_secret_ach}
	For the secret key rate with maximum secrecy $\delta$ and average error probability $\varepsilon$ it holds that
	\begin{equation}
		\widetilde{R}_{max}^*(n,\varepsilon,\delta) \geq \widetilde{C}_S - \sqrt{\frac{V_1}{n}} Q^{-1}(\varepsilon) - \sqrt{\frac{V_2}{n}}Q^{-1}(\delta) + \mathcal{O}\left(\frac{\log(n)}{n}\right) \enspace ,
	\end{equation}
	where
	\begin{equation}\label{eq:v_1_sec_key_rate}
		V_1 = \sum_{\substack{x\in \cX \\ y \in \cY}} P_{XY}(x,y) \log_2^2\left( \frac{P_{XY}(x,y)}{\frac{1}{|\cX|}P_Y(y)} \right) - D\left(P_{XY} || P^{unif}_{\cX} P_Y\right)^2
	\end{equation}
	and
	\begin{equation}\label{eq:v_2_sec_key_rate}
		V_2 = \sum_{\substack{x\in \cX \\ z \in \cZ}} P_{XZ}(x,z) \log_2^2\left( \frac{P_{XZ}(x,z)}{\frac{1}{|\cX|}P_Z(z)} \right) - D\left(P_{XZ} || P^{unif}_{\cX}P_Z\right)^2
	\end{equation}
	with $P^{unif}_{\cX}$ denoting the uniform distribution over the input alphabet $\cX$.
	This rate can be achieved using one-way communication over the public channel.
\end{theorem}
\begin{proof}
	As in the achievability proof sketch for Theorem~\ref{th:secret_key_capacity}, we use a secrecy codebook for the degraded wiretap channel sampled according to a distribution $P_U$, where in this particular case we define $P_U$ to be uniform over the alphabet $\cX$. We mask the codeword symbolwise by computing $u_i + x_i$ using the common randomness provided by the source $P_{XYZ}$ and send the result over the public channel to Bob. Again the artificially created wiretap channel with input $U$ and outputs $(U+X,Y)$ and $(U+X,Z)$ can be used to show achievability results for secret sharing using common randomness, this time in the finite blocklength regime. To compute the achievable secret key rates we apply Theorem~\ref{th:finite_wiretap}. Notice that the secrecy definitions for (degraded) wiretap channels and secret sharing using common randomness match each other. Hence, the secrecy condition of the secrecy code implies the secrecy definition for secret sharing using common randomness.
	
	For the first channel dispersion term $V_1$ the legitimate user channel from $U$ to $(U+X,Y)$ is relevant. Hence, we have
	\begin{align}
		V_1 = \sum_{u\in \cX} P_U(u) \Bigg(&\sum_{\substack{\tilde{x} \in \cX \\ y \in \cY}} P_{U \oplus X,Y|U}(\tilde{x},y|u)\log_2^2\left(\frac{P_{U\oplus X,Y|U}(\tilde{x},y|u)}{P_{U \oplus X,Y}(\tilde{x},y)} \right) \nonumber \\&- D\left(P_{U \oplus X,Y|U=u}||P_{U\oplus X,Y}\right)^2\Bigg) \enspace .
	\end{align}
	Next we analyze the terms $P_{U\oplus X,Y|U}(\tilde{x},y|u)$ and $P_{U\oplus X,Y}(\tilde{x},y)$.
	We have
	\begin{equation}
		P_{U\oplus X,Y|U}(\tilde{x},y|u) = P_{U\oplus X|U}(\tilde{x}|u) P_{Y|U\oplus X,U}(y|\tilde{x},u) = P_X(\tilde{x}-u) P_{Y|X}(y|\tilde{x}-u)
	\end{equation}
	because $U$ is uniformly distributed over $\cX$ and independent of $X,Y$.
	
	Furthermore, it holds that
	\begin{equation}
		P_{U\oplus X,Y}(\tilde{x},y) = P_{U\oplus X}(\tilde{x}) P_{Y|U\oplus X}(y|\tilde{x}) = \frac{1}{|\cX|} P_Y(y)
	\end{equation}
	again because $U$ is uniformly distributed over $\cX$ and independent of $X,Y$. Notice that $U\oplus X$ is independent of $Y$ even though $X$ and $Y$ may not be independent. This is basically because $X$ is encrypted by a one-time pad using $U$ as its key.
	
	Equation~\eqref{eq:v_1_sec_key_rate} follows because the sum in the bracket goes over the entire alphabet sets $\cX$ and $\cY$. Notice that the sums do not depend on the choice of $U$ in the outer sum. The same holds for the sums defining the divergence term.
	
	An analogous argument holds for $V_2$ in Equation~\eqref{eq:v_2_sec_key_rate}.
\end{proof}

\begin{theorem}\label{th:finite_secrecy_conv}
	Let us assume that we use the same communication protocol as in the proof of Theorem~\ref{th:finite_secret_ach}.
	In this case for the secret key rate with average secrecy $\delta$ and average error proabability $\varepsilon$ it holds that
	\begin{equation}
		\widetilde{R}^*_{avg} \leq \widetilde{C}_S - \sqrt{\frac{V_c}{n}} Q^{-1}(\varepsilon + \delta) + \mathcal{O}\left( \frac{\log(n)}{n} \right)\enspace ,
	\end{equation}
	where
	\begin{equation}
		V_c = \sum_{\substack{x\in \cX \\ y\in \cY \\ z\in \cZ}} P_{XYZ}(x,y,z) \log_2^2\left(\frac{P_{XYZ}(x,y,z)}{P_{XZ}(x,z)P_{Y|Z}(y|z)}\right) - D(P_{XYZ}||P_{Y|Z}P_{XZ})^2 \enspace .
	\end{equation}
	This upper bound holds in particular for the secret sharing problem using only one-way communication over the public channel.
\end{theorem}
\begin{proof}
	As mentioned in the theorem, our goal is again to use a wiretap code for the channel with uniformly sampled input $U$ and with outputs $(U+X,Y)$ and $(U+X,Z)$. Hence, an upper bound on the secrecy rate for a wiretap code of this channel is also an upper bound on the achievable secret key rate of the communication protocol for the secret sharing problem.
	
	Using the upper bound for wiretap codes in Theorem~\ref{th:finite_wiretap} it holds that
	\begin{align}\label{eq:conv_finite_help}
		V_c = \sum_{u\in \cX} P_U(u) \Bigg(&\sum_{\substack{\tilde{x}\in \cX \\ y\in \cY \\ z\in \cZ}} P_{U\oplus X,Y,Z|U}(\tilde{x},y,z|u) \log_2^2\left(\frac{P_{U\oplus X,Y,Z|U}(\tilde{x},y,z|u)}{P_{U\oplus X,Z|U}(\tilde{x},z|u) P_{U\oplus X,Y|U\oplus X,Z}(\tilde{x},y|\tilde{x},z)}\right) \nonumber \\
		&-D(P_{U\oplus X,Y,Z|U=u}||P_{U\oplus X,Y|U\oplus X,Z}P_{U\oplus X,Z|U=u})^2\Bigg) \enspace .
	\end{align}
	To simplify Equation~\eqref{eq:conv_finite_help} we observe
	\begin{align}
		P_{U\oplus X,Y,Z|U}(\tilde{x},y,z|u) &= P_{U\oplus X|U}(\tilde{x}|u)P_{Y|U\oplus X,U}(y|\tilde{x},u) P_{Z|Y,U\oplus X,U}(z|y,\tilde{x},u) \nonumber \\
		&= P_X(\tilde{x}-u)P_{Y|X}(y|\tilde{x}-u) P_{Z|Y,X}(z|y,\tilde{x}-u) \enspace ,
	\end{align}
	\begin{equation}
		P_{U\oplus X,Z|U}(\tilde{x},z|u) = P_{U\oplus X|U}(\tilde{x}|u) P_{Z|U\oplus X,U}(z|\tilde{x},u) = P_X(\tilde{x}-u) P_{Z|X}(z|\tilde{x}-u)
	\end{equation}
	and
	\begin{equation}
		P_{U\oplus X,Y|U\oplus X,Z}(\tilde{x},y|\tilde{x},z) = P_{Y|Z}(y|z) \enspace .
	\end{equation}
	The basic idea to show all of these simplifications is to use the facts that U is uniformly distributed over $\cX$ and independent of $X,Y,Z$ and that $U\oplus X$ is independent of $Y$ and $Z$ because $X$ is encrypted by a one-time pad using the uniformly distributed $U$ as a key.
	
	As in the proof of Theorem~\ref{th:finite_secret_ach} the inner sum in Equation~\eqref{eq:conv_finite_help} goes over the entire alphabet sets $\cX,\cY$ and $\cZ$. Hence, the sums do not depend on the specific value of $U$ since it only leads to an index shift within the sum. The same argument holds for the divergence term.
\end{proof}

\begin{corollary}\label{cor:tightness_secret_sharing}
	Let us consider a source of common randomness following a joint pmf $P_{XYZ}$ for which it holds that $X\fooA Y\fooA Z$. Then it holds that the upper bound on the maximal secrecy rate for wiretap codes in Theorem~\ref{th:finite_wiretap} is not tight for the degraded wiretap channel formed by the uniformly sampled input $U$ (over $\cX$), legitimate user output $(U\oplus X,Y)$ and eavesdropper output $(U\oplus X,Z)$.
\end{corollary}
\begin{proof}
	As already emphasized in Remark~\ref{rem:secret_sharing} the communication protocol over the public channel achieving the secret key rate in Theorem~\ref{th:hayashi_secret_key} requires two-way communication over the public channel while the construction of the upper bound in Theorem~\ref{th:finite_secrecy_conv} only requires one-way communication. Hence, in order to prove the corollary, it suffices to show that the upper bound in Theorem~\ref{th:finite_secrecy_conv} exceeds the maximal secret key rate in Theorem~\ref{th:hayashi_secret_key}. We will show in the following that this indeed holds. 
	
	The difference between the rates in Theorems~\ref{th:finite_secrecy_conv} and \ref{th:hayashi_secret_key} lies in the channel dispersion terms $V_c$ and $V_c'$, respectively.
	
	Notice that the sums involving the $\log_2^2$ terms are equivalent. Hence, we focus on the divergence terms in both equations.
	We have
	\begin{align}\label{eq:corollary_proof_1}
		D(P_{XYZ}||P_{Y|Z}P_{XZ})^2 &= \left(\sum_{\substack{x\in \cX \\ y\in \cY \\ z\in \cZ}} P_{XYZ}(x,y,z) \log_2\left(\frac{P_{XYZ}(x,y,z)}{P_{Y|Z}(y|z)P_{XZ}(x,z)}\right)\right)^2 \nonumber \\
		&=\left( \sum_{x\in \cX} P_X(x) \sum_{\substack{y\in \cY \\ z\in \cZ}} P_{YZ|X}(y,z|x) \log_2\left(\frac{P_{YZ|X}(y,z|x)}{P_{Y|Z}(y|z)P_{Z|X}(z|x)}\right) \right)^2
	\end{align}
	and
	\begin{align}\label{eq:corollary_proof_2}
		&\sum_{x\in \cX} P_X(x) D(P_{YZ|X=x}||P_{Y|Z}P_{Z|X=x})^2 \nonumber \\ = &\sum_{x\in \cX} P_{X}(x) \left(\sum_{\substack{y\in \cY \\ z\in \cZ}} P_{YZ|X}(y,z|x) \log_2(\frac{P_{YZ|X}(y,z|x)}{P_{Y|Z}(y|z)P_{Z|X}(z|x)})\right)^2 \enspace .
	\end{align}
	
	Because $f(x)=x^2$ is a strictly convex function we have that
	\begin{equation}
		D(P_{XYZ}||P_{Y|Z}P_{XZ})^2 > \sum_{x\in \cX} P_{X}(x) D(P_{YZ|X=x}||P_{Y|Z}P_{Z|X=x})^2
	\end{equation}
	and hence the statement follows.
\end{proof}
\begin{remark} \label{rem:tightness_upper_hayashi}
	Corollary~\ref{cor:tightness_secret_sharing} shows for sources of common randomness $P_{XYZ}$ for which $X \fooA Y \fooA Z$ that the upper bound obtained in Theorem~\ref{th:finite_secrecy_conv} is strictly larger than the rate obtained in Theorem~\ref{th:hayashi_secret_key}. The secret key rate of Theorem~\ref{th:hayashi_secret_key} also provides an upper bound on the secret key rate for protocols with one-way communication. Notice though that Theorem~\ref{th:hayashi_secret_key} requires the random variables $X,Y,Z$ to be connected via a Markov chain while for Theorem~\ref{th:finite_secrecy_conv} can also be applied if this is not the case.
\end{remark}

Equipped with the knowledge of this section we next tackle the problem of securing the HDA against the attack scenarios mentioned in Section~\ref{subsec:attacker_model} for finite lengths.

\section{Achievability and Converse Bounds for Finite Lengths}\label{sec:finite}
The results presented in the previous section are for the asymptotic setting in the sense that they hold if the number of capacitive cells in the PUF goes to infinity and the fraction of the erasures for the attacker are determined by the erasure channel parameters $p_d$ and $p_a$ for the digital and analog attacker model. Of course, as a first order approximation one could compute the key capacity $C^q_{key}$ and estimate the dimension of the secrecy code for the resulting (degraded) wiretap channel by $C^q_{key} n$. As for the channel models discussed in Section~\ref{sec:finite_blocklength} to estimate the code dimension in that way is highly inaccurate if the blocklength is of small to moderate size. In general the achievable code dimension is of significantly smaller size. We note that the erasure channel assumption is only an approximation of the reality since for our attacker model the number of erasures is assumed to be constant and equal to $p_d n$ or $p_a n$ for the digital and the analog attacker model, respectively. Nevertheless, we consider the results obtaining from this to be rather accurate because these values are equal to the expectation for the erasure channels. Notice here that erasures in the digital attacker model imply erasures in the analog attacker model. The erasures for the analog attacker are then added on top.

Using Theorem~\ref{th:finite_secret_ach} to estimate the achievable code rate leads to a much more precise estimate for small to moderate blocklengths. Similarly, applying Theorem~\ref{th:finite_secrecy_conv} is much more precise to estimate a converse result on the secret key rate for the HDA. However, using Theorem~\ref{th:hayashi_secret_key} enables us to find an even tighter converse in case we only need to protect the design against the digital attacker.

The following observation is helpful and used implicitly throughout this section.
Empirically, we found that $I(S;\widetilde{S}|W) \approx I(S;\widetilde{S})$. Basically we observed no difference between these quantities even though we were not able to formally prove this statement. Notice that this does not mean that the helper data has not been utilized on the right hand side of the equation. Rather it means that the channel matrix on the right hand side is constructed by averaging over all possible values of $W$ rather than constructing the channel matrix for each realization $w$, computing $I(S;\widetilde{S}|W=w)$ and then averaging with respect to $f_W(w)$.
This effect is desirable as it shows that it is sufficient to construct a single codebook independent of the helper data $W$, rather than constructing different codebooks for different values of $W$. This behaviour makes perfect sense considering that $\widetilde{S}$ is an approximation of $S$ and $W$ is independent of $S$ due to the zero leakage condition (see Section~\ref{subsec:zero_leakage_hds}).

The potential inaccuracy by assuming $I(S;\widetilde{S}) = I(S;\widetilde{S}|W)$ only induces a rate penalty for the scheme. In terms of security this is not a problem as $W$ is still perfectly utilized and the codebook is chosen by the legitimate users. Furthermore, it makes the application of Theorem~\ref{th:finite_wiretap} significantly easier.

In the following, we are investigating which finite rates can be achieved by HDAs with different parameters for both attacker models. We are in particular interested in this rate as the code dimension in the HDA has to be at least as large as the targeted security level. A lower bound on the code rate (which also depends on the blocklength for finite $n$) leads to an upper bound on the required amount of capacitive cells while a lower bound can be used to prove impossibility results, i.e. it enables to show that a certain amount of cells is absolutely required for target values in terms of reliability and security of the PUF.

As in Section~\ref{sec:securing_foil} we first consider the HDA scheme only covering the digital attacker.

\subsection{Digital Attacker}
\begin{theorem}\label{th:finite_hda_ach}
	The maximal achievable rate $\widetilde{R}^{key,dig}_{max}(n,\varepsilon,\delta)$ for the HDA achieving a maximum secrecy level $\delta$ against the digital attacker described in Section~\ref{subsec:attacker_model} with error probability $\varepsilon$ during regular device operation is lower bounded by
	\begin{equation}\label{eq:finite_hda_ach}
		\widetilde{R}^{key,dig}_{max}(n,\varepsilon,\delta) \geq R^{q}_{asymp,dig} - \sqrt{\frac{V_1}{n}} Q^{-1}(\varepsilon) - \sqrt{\frac{V_2}{n}} Q^{-1}(\delta) + \mathcal{O}\left(\frac{\log(n)}{n}\right) \enspace ,
	\end{equation} 
	where
	\begin{equation}
		R^{q}_{asymp,dig} \defeq  I(S;\widetilde{S})p_d \enspace ,
	\end{equation}
	\begin{equation}
		V_1 = \sum_{\substack{s \in \cS \\ \tilde{s} \in \cS}} P_{S,\widetilde{S}}(s,\tilde{s}) \log_2^2\left(\frac{P_{S,\widetilde{S}}(s,\tilde{s})}{\frac{1}{|\cS|} P_{\widetilde{S}}(\tilde{s})}\right) - D(P_{S,\widetilde{S}}||P_\cS^{unif} P_{\widetilde{S}})^2
	\end{equation}
	and
	\begin{align}
		V_2 = &\sum_{\substack{s\in \cS \\ \tilde{s}\in \cS}} P_{S,\widetilde{S}}(s,\tilde{s}) (1-p_d) \log_2^2\left(\frac{P_{S,\widetilde{S}}(s,\tilde{s})}{\frac{1}{|\cS|}P_{\widetilde{S}}(\tilde{s})}\right) + \sum_{s\in \cS} P_S(s) p_d \log_2^2\left(\frac{P_S(s)}{\frac{1}{|\cS|}}\right) \nonumber \\
		&- \left[(1-p_d) D\left(P_{S,\widetilde{S}}||P_{\cS}^{unif}P_{\widetilde{S}}\right) + p_d D\left(P_S||P_{\cS}^{unif}\right)\right]^2 \enspace .
	\end{align}
\end{theorem}
\begin{proof}
	The statement follows from Theorem~\ref{th:finite_secret_ach} by setting $S\equiv X$, $\widetilde{S} \equiv Y$ and $\widetilde{S}_d \equiv Z$. $V_1$ is directly obtained by plugging in the respective random variables. For the computation of $V_2$ we note that the probability that $P_{\widetilde{S}_d|S}(E|s)=p_d$ irrespective of $s$. In case no erasure occurs it holds that $P_{\widetilde{S}_d|S}(\tilde{s}|s) = P_{\widetilde{S}|S}(\tilde{s}|s) (1-p_d)$.
	The rest follows easily from the definition of $V_2$ in Theorem~\ref{th:finite_secret_ach}.
\end{proof}

\begin{theorem}\label{th:upper_bound_finite_hda}
	The secret key rate for a source of common randomness with distribution  $P_{S,\widetilde{S},\widetilde{S}_d}$ is
	\begin{equation}
		\widetilde{R}^*_{avg}(n,\varepsilon,\delta) = C_S - \sqrt{\frac{V_c'}{n}} Q^{-1}(\varepsilon + \delta) + \mathcal{O}\left( \frac{\log(n)}{n} \right) \enspace ,
	\end{equation}
	where
	\begin{equation}\label{eq:v_c'_digital}
		V_c' = \sum_{s\in \cS} P_S(s) \sum_{\tilde{s} \in \cS} P_{\widetilde{S}|S}(\tilde{s}|s) p_d \log_2^2\left(\frac{P_{\widetilde{S}|S}(\tilde{s}|s)}{P_{\widetilde{S}}(\tilde{s})}\right) - \sum_{s\in \cS} P_S(s) p_d^2 D\left(P_{\widetilde{S}|S=s}||P_{\widetilde{S}}\right)^2 \enspace .
	\end{equation}
	Furthermore, it holds that
	\begin{equation}
		\widetilde{R}^*_{avg}(n,\varepsilon,\delta) \geq \widetilde{R}^{key}_{avg}(n,\varepsilon,\delta) \enspace ,
	\end{equation}
	where $\widetilde{R}^{q}_{avg,dig}(n,\varepsilon,\delta)$ denotes the maximal achievable rate for the HDA achieving an average security level $\delta$ against the digital attacker 
\end{theorem}
\begin{proof}
	The statement follows from Theorem~\ref{th:hayashi_secret_key} by setting $S\equiv X$, $\widetilde{S} \equiv Y$ and $\widetilde{S}_d \equiv Z$. Notice that $S,\widetilde{S},\widetilde{S}_d$ form the Markov chain $S \fooA \widetilde{S} \fooA \widetilde{S}_d$ which is necessary for Theorem~\ref{th:hayashi_secret_key} to be applicable.
	
	By the definition of $\widetilde{S}_d$ we have that $\widetilde{S}_d$ is either equal to $\widetilde{S}$ or the erasure event $E$. We have that
	\begin{align}
		P_{S,\widetilde{S},\widetilde{S}_d}(s,\tilde{s},\tilde{s}) &= P_S(s)P_{\widetilde{S}|S}(\tilde{s}|s)(1-p_d) \nonumber \\ P_{S,\widetilde{S},\widetilde{S}_d}(s,\tilde{s},E) &= P_S(s)P_{\widetilde{S}|S}(\tilde{s}|s) p_d.
	\end{align}
	
	Furthermore, it holds that $P_{\widetilde{S}|\widetilde{S}_d}(\tilde{s}|\tilde{s}) = 1$ and $P_{\widetilde{S}|\widetilde{S}_d}(\tilde{s}|E)=P_{\widetilde{S}}(\tilde{s})$.
	
	Using these statements and applying them to Theorem~\ref{th:hayashi_secret_key} we observe that all terms for which $\widetilde{S}_d\neq E$ go to zero.
	
	Equation~\eqref{eq:v_c'_digital} follows then by considering the remaining terms for which $\widetilde{S}_d=E$.
\end{proof}

By Remark~\ref{rem:tightness_upper_hayashi} we know that the upper bound given in Theorem~\ref{th:upper_bound_finite_hda} is tighter than applying Theorem~\ref{th:finite_secrecy_conv}.

Using Theorem~\ref{th:finite_hda_ach} and Theorem~\ref{th:upper_bound_finite_hda} we obtain the achievability and converse results presented in Tables~\ref{tab:finite_rates_0.1_digital}, \ref{tab:achievable_rates_0.18_digital}, \ref{tab:achievable_rates_0.1_e9} and \ref{tab:achievable_rates_0.2_e9}. When we speak of a $\lambda$ Bit security level we mean that the security parameter $\delta$ satisfies $\delta \leq 2^{-\lambda}$.

\begin{table}[H]
	\caption{Achievability (ach.) and converse (conv.) results on the number of necessary capacitive PUF cells for $p_d=0.1$, PUF reliability $\varepsilon=10^{-6}$ and security levels 128, 192 and 256 bit, digital attacker, equiprobable input quantization, $\sigma_P = 2241, \sigma_N = 129$}
	\centering
	\begin{tabular}{|c|c|c|c|c|c|c|}
		\hline
		quantizer & ach.128b & conv.128b & ach.192b & conv.192b & ach.256b & conv.256b\\
		\hline
		2 & 3645 & 1902 & 5499 & 2655 & 7354 & 3391\\
		4 & 2664 & 1117 & 4025 & 1516 & 5386 & 1902\\
		8 & 3178 & 850 & 4635 & 1128 & 6072 & 1398\\
		16 & 5502 & 779 & 7768 & 1019 & 9977 & 1250\\
		32 & 5390 & 744 & 7609 & 970 & 9773 & 1187\\
		64 & 5509 & 726 & 7768 & 943 & 9968 & 1152\\
		\hline
	\end{tabular}
	\label{tab:finite_rates_0.1_digital}
\end{table}
\begin{table}[H]
	\caption{Achievability (ach.) and converse (conv.) results on the number of necessary capacitive PUF cells for $p_d=0.18$, PUF reliability $\varepsilon=10^{-6}$ and security levels 128, 192 and 256 bit, digital attacker, equiprobable input quantization, $\sigma_P = 2241, \sigma_N = 129$}
	\centering
	\begin{tabular}{|c|c|c|c|c|c|c|}
		\hline
		quantizer & ach.128b & conv.128b & ach.192b & conv.192b & ach.256b & conv.256b\\
		\hline
		2 & 1938 & 1038 & 2923 & 1454 & 3909 & 1860\\
		4 & 1399 & 606 & 2113 & 825 & 2828 & 1038\\
		8 & 1508 & 459 & 2216 & 612 & 2916 & 760\\
		16 & 2194 & 420 & 3128 & 552 & 4042 & 680\\
		32 & 2150 & 401 & 3064 & 525 & 3959 & 644\\
		64 & 2179 & 390 & 3102 & 510 & 4004 & 625\\
		\hline
	\end{tabular}
	\label{tab:achievable_rates_0.18_digital}
\end{table}
\begin{table}[H]
\caption{Achievability (ach.) and converse (conv.) results on the number of necessary capacitive PUF cells for $p=0.1$, PUF reliability $\varepsilon=10^{-9}$ and security levels 128, 192 and 256 bit, digital attacker, equiprobable input quantization, $\sigma_P = 2241, \sigma_N = 129$}
	\centering
	\begin{tabular}{|c|c|c|c|c|c|c|}
		\hline
		quantizer & ach.128b & conv.128b & ach.192b & conv.192b & ach.256b & conv.256b\\
		\hline
		2 & 3645 & 2106 & 5499 & 2887 & 7354 & 3647\\
		4 & 2665 & 1286 & 4026 & 1703 & 5388 & 2106\\
		8 & 3345 & 1006 & 4834 & 1300 & 6299 & 1582\\
		16 & 6050 & 940 & 8414 & 1194 & 10705 & 1437\\
		32 & 5927 & 903 & 8243 & 1142 & 10488 & 1370\\
		64 & 6068 & 884 & 8427 & 1114 & 10712 & 1335\\
		\hline
	\end{tabular}
	\label{tab:achievable_rates_0.1_e9}
\end{table}

\begin{table}[H]
	\caption{Achievability (ach.) and converse (conv.) results on the number of necessary capacitive PUF cells for $p=0.18$, PUF reliability $\varepsilon=10^{-9}$ and security levels 128, 192 and 256 bit, digital attacker, equiprobable input quantization, $\sigma_P = 2241, \sigma_N = 129$}
	\centering
	\begin{tabular}{|c|c|c|c|c|c|c|}
		\hline
		quantizer & ach.128b & conv.128b & ach.192b & conv.192b & ach.256b & conv.256b\\
		\hline
		2 & 1938 & 1145 & 2923 & 1575 & 3909 & 1994\\
		4 & 1400 & 693 & 2114 & 923 & 2828 & 1145\\
		8 & 1571 & 539 & 2291 & 701 & 3002 & 856\\
		16 & 2382 & 504 & 3350 & 643 & 4293 & 777\\
		32 & 2334 & 483 & 3282 & 614 & 4205 & 740\\
		64 & 2370 & 472 & 3327 & 599 & 4259 & 720\\
		\hline
	\end{tabular}
	\label{tab:achievable_rates_0.2_e9}
\end{table}

The results show that different amounts of input quantization intervals can significantly change the achievable rates for the foil PUF. Notice that for the achievability bounds, increasing the amount of quantization intervals does not necessarily decrease the required number of PUF cells. 2 Bit quantization is optimal from that perspective, e.g. for the setting examined in Table~\ref{tab:finite_rates_0.1_digital} the results show that $2664$ PUF cells suffice for a security level of 128 Bit. The converse bounds decrease with an increasing the amount of quantization intervals though. From the tables we observe that an increased demand in PUF reliability only has a small influence on the required number of cells, whereas the erasure probability has a much larger impact. We also observe that higher security levels require more PUF cells as was to be expected.

\subsection{Analog Attacker}
In this section we basically perform a similar analysis for the analog attacker as we did for the digital attacker.

\begin{theorem}\label{th:ach_analog_finite}
	The maximal achievable rate $\widetilde{R}^{key,ana}_{max}(n,\varepsilon,\delta)$ for the HDA achieving a maximum secrecy level $\delta$ against the analog attacker described in Section~\ref{subsec:attacker_model} with error probability $\varepsilon$ during regular device operation is lower bounded by
	\begin{equation}\label{eq:finite_hda_ach_analog}
		\widetilde{R}^{key,ana}_{max}(n,\varepsilon,\delta) \geq R^{q,lower}_{asymp,ana} - \sqrt{\frac{V_1}{n}} Q^{-1}(\varepsilon) - \sqrt{\frac{V_2}{n}} Q^{-1}(\delta) + \mathcal{O}\left(\frac{\log(n)}{n}\right) \enspace ,
	\end{equation} 
	where
	\begin{equation}
		R^{q,lower}_{asymp,ana} \defeq I(S;\widetilde{S}|W) (1-p_a+p_d) - H(S)(1-p_a) \enspace ,
	\end{equation}
	\begin{equation}
		V_1 = \sum_{\substack{s \in \cS \\ \tilde{s} \in \cS}} P_{S,\widetilde{S}}(s,\tilde{s}) \log_2^2\left(\frac{P_{S,\widetilde{S}}(s,\tilde{s})}{\frac{1}{|\cS|} P_{\widetilde{S}}(\tilde{s})}\right) - D(P_{S,\widetilde{S}}||P_\cS^{unif} P_{\widetilde{S}})^2
	\end{equation}
	and
	\begin{align}
		V_2 = & p_d \sum_{s\in \cS} P_S(s) \log_2^2\left(\frac{P_S(s)}{\frac{1}{|\cS|}}\right) + (p_a-p_d) \sum_{\substack{s\in \cS \\ \tilde{s}\in \cS}} P_{S,\widetilde{S}}(s,\tilde{s}) \log_2^2\left(\frac{P_S(s)P_{\widetilde{S}|S}(\tilde{s}|s)}{\frac{1}{|\cS|}P_{\widetilde{S}}(\tilde{s})}\right) \nonumber \\
		&+ (1-p_a) \log_2^2(|\cS|) - \left[\log_2(|\cS|) + I(S;\widetilde{S})(p_a-p_d) - p_a H(S)\right]^2
	\end{align}
\end{theorem}
\begin{proof}
	To prove the statement we use Theorem~\ref{th:finite_secret_ach} and set $S\equiv X$, $\widetilde{S} \equiv Y$ and $(\widetilde{S}_d,\widetilde{S}_a) \equiv Z$.
	The channel for the legitimate user does not change compared to the digital attacker. Hence, the dispersion term $V_1$ does not change compared to Theorem~\ref{th:finite_hda_ach}.
	
	Therefore, we only need to provide a proof for $V_2$.
	For $\widetilde{S}_d \neq E$ and $\widetilde{S}_a \neq E$, we have that
	\begin{align}
		P_{S,\widetilde{S}_d,\widetilde{S}_a}(s,\tilde{s},s) &= P_{S,\widetilde{S}}(s,\tilde{s}) (1-p_a) \nonumber \\
		P_{\widetilde{S}_d,\widetilde{S}_a}(\tilde{s},s) &= P_S(s) P_{\widetilde{S}|S}(\tilde{s}|s) (1-p_a) \enspace .
	\end{align}
	
	Furthermore, for $\widetilde{S}_d = E$ and $\widetilde{S}_a = E$
	\begin{align}
		P_{S,\widetilde{S}_d,\widetilde{S}_a}(s,E,E) &= P_{S}(s) p_d \nonumber \\
		P_{\widetilde{S}_d,\widetilde{S}_a}(E,E) &= p_d
	\end{align}
	
	and for $\widetilde{S}_d \neq E$ and $\widetilde{S}_a = E$
	\begin{align}
		P_{S,\widetilde{S}_d,\widetilde{S}_a}(s,\tilde{s},E) &= P_{S}(s) P_{\widetilde{S}|S}(\tilde{s}|s) (p_a-p_d) \nonumber \\
		P_{\widetilde{S}_d,\widetilde{S}_a}(\tilde{s},E) &= P_{\widetilde{S}}(\tilde{s}) (p_a-p_d)
	\end{align}
	
	Using these identities and Theorem~\ref{th:finite_secret_ach} we obtain
	\begin{align}
		V_2 = & p_d \sum_{s\in \cS} P_S(s) \log_2^2\left(\frac{P_S(s)}{\frac{1}{|\cS|}}\right) + (p_a-p_d) \sum_{\substack{s\in \cS \\ \tilde{s}\in \cS}} P_{S,\widetilde{S}}(s,\tilde{s}) \log_2^2\left(\frac{P_S(s)P_{\widetilde{S}|S}(\tilde{s}|s)}{\frac{1}{|\cS|}P_{\widetilde{S}}(\tilde{s})}\right) \nonumber \\
		&+ (1-p_a) \log_2^2(|\cS|) - \left[\log_2(|\cS|) + I(S;\widetilde{S})(p_a-p_d) - p_a H(S)\right]^2
	\end{align}
	after some elementary algebraic steps.
\end{proof}

\begin{theorem}\label{th:conv_analog_finite}
	The maximal achievable rate $\widetilde{R}^{key,ana}_{avg}(n,\varepsilon,\delta)$ for the HDA achieving a average secrecy level $\delta$ against the analog attacker described in Section~\ref{subsec:attacker_model} with error probability $\varepsilon$ during regular device operation is upper bounded by
	\begin{equation}
		\widetilde{R}^{key,ana}_{avg}(n,\varepsilon,\delta) \leq R^{q,upper}_{asymp,ana} - \sqrt{\frac{V_c}{n}} Q^{-1}(\varepsilon + \delta) + \mathcal{O}\left(\frac{\log(n)}{n}\right)
	\end{equation}
	with
	\begin{equation}
		R^{q,upper}_{asymp,ana} \defeq I(S;\widetilde{S}|W) p_d
	\end{equation}
	and
	\begin{equation}
		V_c = p_d \sum_{\substack{s\in \cS \\ \tilde{s} \in \cS}} P_{S,\widetilde{S}}(s,\tilde{s}) \log_2^2\left(\frac{P_{\widetilde{S}|S}(\tilde{s}|s)}{P_{\widetilde{S}}(\tilde{s})}\right) - p_d^2 I(S;\widetilde{S})^2 \enspace .
	\end{equation}
\end{theorem}
\begin{proof}
	To prove the statement we use Theorem~\ref{th:finite_secrecy_conv} and set $S \equiv X$, $\widetilde{S} \equiv Y$ and $(\widetilde{S}_d,\widetilde{S}_a) \equiv Z$.
	For $\widetilde{S}_d \neq E$ and $\widetilde{S}_a \neq E$ we have that
	\begin{align}
		P_{S,\widetilde{S},\widetilde{S}_d,\widetilde{S}_a}(s,\tilde{s},\tilde{s},s) &= P_{S,\widetilde{S}}(s,\tilde{s}) (1-p_a) \nonumber \\
		P_{S,\widetilde{S}_d,\widetilde{S}_a}(s,\tilde{s},s) &= P_S(s) P_{\widetilde{S}|S}(\tilde{s}|s)(1-p_a) \nonumber \\
		P_{\widetilde{S}|\widetilde{S}_d,\widetilde{S}_a}(\tilde{s}|\tilde{s},s) &= 1 \enspace ,
	\end{align}
	for $\widetilde{S}_d = E$ and $\widetilde{S}_a = E$ we have that
	\begin{align}
		P_{S,\widetilde{S},\widetilde{S}_d,\widetilde{S}_a}(s,\tilde{s},E,E) &= P_{S,\widetilde{S}}(s,\tilde{s}) p_d \nonumber \\
		P_{S,\widetilde{S}_d,\widetilde{S}_a}(s,E,E) &= P_S(s) p_d \nonumber \\
		P_{\widetilde{S}|\widetilde{S}_d,\widetilde{S}_a}(\tilde{s}|E,E) &= P_{\widetilde{S}}(\tilde{s})
	\end{align}
	and for $\widetilde{S}_d \neq E$ and $\widetilde{S}_a = E$ it holds that
	\begin{align}
		P_{S,\widetilde{S},\widetilde{S}_d,\widetilde{S}_a}(s,\tilde{s},\tilde{s},E) &= P_S(s) P_{\widetilde{S}|S}(\tilde{s}|s) (p_a-p_d) \nonumber \\
		P_{S,\widetilde{S}_d,\widetilde{S}_a}(s,\tilde{s},s) &= P_S(s) p_d \nonumber \\
		P_{\widetilde{S}|\widetilde{S}_d,\widetilde{S}_a}(\tilde{s}|\tilde{s},s) &= P_{\widetilde{S}}(\tilde{s})
	\end{align}
	Using Theorem~\ref{th:finite_secrecy_conv} and the identities outlined above directly gives the desired result after some elementary steps.
\end{proof}
\begin{remark}
	Notice that it is to be assumed that the converse bound is of poor quality as the analog erasure probability $p_a$ has no influence on the bound. The approach taken for the digital attacker to use Theorem~\ref{th:hayashi_secret_key} to find an upper bound on the maximal rate of the HDA cannot be done for the analog attacker in a straightforward manner. The reason is that Theorem~\ref{th:hayashi_secret_key} is only applicable if $S$, $\widetilde{S}$ and $(\widetilde{S}_d,\widetilde{S}_a)$ form a Markov chain, which is not the case.
\end{remark}

The results for applying Theorem~\ref{th:ach_analog_finite} and Theorem~\ref{th:conv_analog_finite} to the HDA in the analog attacker scenario are presented in
Table~\ref{tab:achievable_rates_0.18_e6_analog_equiprob} and Table~\ref{tab:achievable_rates_0.18_e6_analog_equidist}. Empty spots in the tables signal that the amount of required PUF cells is above 20000 and the respective parameter sets therefore have been considered impractical for implementations due to better alternatives. Achievability as well as converse results on the required number of PUF cells are given. Comparing the achievability results to the results for the digital attacker, we observe that for a small amount of quantization levels are almost identical. Even though we cannot say the same for the converse results this is an interesting observation. Our results suggest that it is better to use coarse input quantizers rather than fine ones. Furthermore, for a very low number of quantization intervals equiprobable quantization performs better than equidistant one. We observe that $4$ quantization levels with equiprobable input quantization give the best achievability bounds for all security levels in this case. As equiprobable quantization has the benefit of leaking no information about the secret via the reconstruction helper data $\widetilde{W}^n$ (do not confuse this with the quantization helper data $W^n$) this is particularly interesting. Aside from requiring less PUF cells from an achievability bound perspective quantizers with fewer levels are cheaper and easier to build. Also notice that the output quantizer levels need to be adjusted according to the helper data. In this respect this observation becomes even more important than in the enrollment phase.

\begin{table}[t]
	\caption{Achievability (ach.) and converse (conv.) results on the number of necessary capacitive PUF cells for $p_d=0.18$ and $p_a=0.36$, PUF reliability $\varepsilon=10^{-6}$ and security levels 128, 192 and 256 bit, analog attacker, equiprobable quantization, $\sigma_P = 2241, \sigma_N = 129$}
	\centering
	\begin{tabular}{|c|c|c|c|c|c|c|}
		\hline
		quantizer & ach.128b & conv.128b & ach.192b & conv.192b & ach.256b & conv.256b\\
		\hline
		2 & 1938 & 902 & 2923 & 1295 & 3909 & 1683\\
		4 & 1399 & 417 & 2113 & 607 & 2828 & 795\\
		8 & 1511 & 239 & 2219 & 358 & 2918 & 478\\
		16 & 5983 & 201 & 8470 & 301 & 10897 & 401\\
		32 & - & 187 & - & 280 & - & 373\\
		64 & - & 179 & - & 269 & - & 358\\
		\hline
	\end{tabular}
	\label{tab:achievable_rates_0.18_e6_analog_equiprob}
\end{table}

\begin{table}[t]
	\caption{Achievability (ach.) and converse (conv.) results on the number of necessary capacitive PUF cells for $p=0.18$ and $p_a=0.36$, PUF reliability $\varepsilon=10^{-6}$ and security levels 128, 192 and 256 bit, analog attacker, equidistant quantization, $\sigma_P = 2241, \sigma_N = 129$}
	\centering
	\begin{tabular}{|c|c|c|c|c|c|c|}
		\hline
		quantizer & ach.128b & conv.128b & ach.192b & conv.192b & ach.256b & conv.256b\\
		\hline
		2 & 1938 & 902 & 2923 & 1295 & 3909 & 1683\\
		4 & 2305 & 740 & 3482 & 1070 & 4659 & 1396\\
		8 & 1679 & 363 & 2537 & 545 & 3397 & 726\\
		16 & 1355 & 245 & 2044 & 367 & 2734 & 490\\
		32 & 2184 & 190 & 3141 & 284 & 4081 & 379\\
		64 & - & 183 & - & 275 & - & 366\\
		\hline
	\end{tabular}
	\label{tab:achievable_rates_0.18_e6_analog_equidist}
\end{table}
\section{Applying Converse Results in Security Analysis}\label{sec:converse_application}
In previous wiretap work \cite{GXKF22}, the complexity of the attacker has been estimated by $H_{att} = - \sum_{i}^{n_s} p_{s,i} \log_2(p_{s,i})$, where $n_s$ is the dimension of the code and $p_{s,i}$ denotes the symbol error rate of an attack on the information symbols of the polar code. For the parameters in the paper this leads to a claimed attack complexity of $100$ bit. However, going to the physical channel, only $n  p_d$ symbols are destroyed. For the discussed use case with $128$ PUF cells, $23$ destroyed nodes and $8$ quantization intervals, this leads to a missing information of $69$ bits depending on the quantization level. When also considering frozen bits in the polar code and additional quantization leakage, the resulting security level is notably under the claimed $100$ bit.

However, there is also another way to show that this scheme cannot achieve the required security level of $100$ bit in the digital attacker scenario. We can simply use our converse bound (Theorem~\ref{th:upper_bound_finite_hda}) and observe that for a security level of $100$ bit and error probability $10^{-6}$ we would require at least 389 PUF cells. Thereby, this construction breaks the converse bound and hence the construction cannot reach the desired security level without substantially leaking information about the secret through the helper data.

\section{Summary and Outlook}\label{sec:conclusion_pufs}

Over the last years, several practical papers for key generation for PUF-based tamper protection were introduced. In this work, we analyzed the problem with information theoretical tools for additional insights and contribute to understanding and quantifying theoretical and practical limits.

We have given achievable rates under the constraint that the quantization helper data leaks no information about the quantized PUF response $S^n$ obtained during the enrollment phase both in the asymptotic as well as the non-asymptotic setting for two different attacker models. Asymptotically we observed that for the digital attacker equiprobable input quantization performs better than equidistant input quantization in terms of achievable code rates. In case the analog attacker of concern to designer equiprobable quantization can be beneficial for larger quantization alphabets. The same observation holds in the finite length regime. Our results have practical merit since they show that coarse quantization is not only easier to implement but also often requires less capacitive cells on the foil. Furthermore, we presented converse bounds that showed that existing implementations cannot achieve a security level of 100 bits if leakage of information via the public quantization helper data is prevented. The achievability results allow to state an upper bound on the required number of capacitive PUF cells. Thereby, they permit to set design goals for PUFs with respect to the required security and reliability levels.

Our results show that we can achieve the number of required capacitive cells with an optimized implementation, e.g. through a feasible feature shrink. Consequently, by quantifying the fundamental limits of this PUF architecture, the results presented in this work contribute to guiding practical work towards theoretically secure and certifiable future implementations.

\subsection*{Acknowledgements}
The authors thank Wei Yang for helpful conversations on his work on wiretap channels in the finite blocklength regime. They would also like to thank Frans Willems for pointing out several references on zero leakage quantization that proved very helpful in the context of this work. Further thanks goes to Antonia Wachter-Zeh for helpful discussions. We would like to thank Frans Willems 
Georg Maringer would like to thank Giuseppe Durisi for hosting him twice for research visits at Chalmers University and in particular for introducing him into the field of finite blocklength information theory on those occasions.
The authors especially thank the anonymous reviewers for their comments that significantly improved the quality of this manuscript.

This research presented in this work was supported by the Bavarian Ministry of Economic Affairs, Regional Development and Energy, in the project Trusted Electronics Bavaria Center.
\bibliographystyle{alpha}
\bibliography{main}

\appendix
\section{Basic Notions of Information and Coding Theory}\label{app:basics_info_coding}
Several notions introduced in this section are standard in information theory. Hence, they can be found in textbooks like \cite{cover1999elements}.

\begin{definition}[Memoryless channel]
	We say that a channel is memoryless if it holds that
	\begin{equation}
		P_{Y^n|X^n}(\mathbf{y}^n|\mathbf{x}^n) = \prod_{i=1}^n P_{Y_i|X_i}(y_i|x_i) \enspace .
	\end{equation}
\end{definition}

\begin{definition}
We say that a channel created an error at the decoder's input if an input symbol $x$ has been mapped to an output $y\neq x \in \cX$. We define a symbol erasure to be the event that an input symbol $x$ has been erased and hence this symbol gives the receiver no information about the channel input. Notice that the synchronization between channel input and channel output is not lost in this case as it is when symbols are deleted. If an erasure occurs at the $i$-th symbol we denote this by setting the random variable $Y_i=E$, where the event $E$ denotes that an erasure occurred.
\end{definition}

\begin{definition}
	We denote the average block error probability by
	\begin{equation}
		\overline{P}_e \coloneqq \frac{1}{|\cM|} \sum_{m \in \cM} P_e(m)\enspace ,
	\end{equation}
	where
	\begin{equation}
		P_e(m) = P_{Y^*|\cM}(\cD(Y^*)\neq m|m)
	\end{equation}
	denotes the probability that an error occurs if the message $m$ is transmitted.
	We denote the maximal error probability by
	\begin{equation}
		P_{e,max} = \max_{m\in \cM} P_e(m) \enspace .
	\end{equation}
\end{definition}

\begin{definition}
We define the \textbf{rate} $R$ of a code of cardinality $M$ and blocklength $n$ to be
\begin{equation}
    R \coloneqq \frac{1}{n} \log_2(M) \enspace .
	\end{equation}
\end{definition}

\begin{definition}
	We define the support of a probability mass function $P_X$ for a random variable $X$ taking values in a set $\cX$ by
	\begin{equation}
		\supp(P_X) \defeq \{x \in \cX: P_X(x)>0 \} \enspace . 
	\end{equation}
\end{definition}
\begin{definition}\label{def:disc_entropy}
	We define the \textbf{entropy} of a discrete RV $X$ to be
	\begin{equation}
		H(X) \coloneqq -\sum_{a \in \supp(P_X)} P_X(a) \log_2 (P_X(a))
	\end{equation}
	and the conditional entropy of $X$ given that some random variable $Z=z$ to be
	\begin{equation}
		H(X|Z=z) \coloneqq -\sum_{a \in \supp(P_{X|Z=z})} P_{X|Z} (a|z) \log_2(P_{X|Z}(a|z)) \enspace .
	\end{equation}
	Furthermore, we denote the conditional entropy of $X$ given $Z$ to be
	\begin{equation}\label{eq:cond_entropy}
		H(X|Z) \coloneqq -\sum_{(a,b) \in \supp(P_{XZ})} P_{XZ}(a,b) \log_2(P_{X|Z}(a|b)) = \sum_{b \in \supp(Z)} P_Z(b) H(X|Z=b) \enspace .
	\end{equation}
\end{definition}

\begin{definition}
	We define the \textbf{differential entropy} of a continuous RV $X$ with probability density function (pdf) $f_X$ by
	\begin{equation}
		h(X) \coloneqq -\int_{a \in \supp{f_X}} f_X(a) \log_2 (f_X(a)) \diff a \enspace .
	\end{equation}
	The definition of the conditional differential entropy of $X$ given some random variable $Z$ is analogous to eq.~\eqref{eq:cond_entropy} in Definition~\ref{def:disc_entropy}.
\end{definition}

\begin{definition}
	We define the \textbf{variational distance} $d(P,Q)$ between two probability mass functions (pmfs) $P$ and $Q$ defined over the same domain $\mathcal{X}$ by
	\begin{equation}
		d(P,Q) = \frac{1}{2} \sum_{x \in \mathcal{X}} \lvert P(x) - Q(x) \rvert \enspace .
	\end{equation}
\end{definition}
\begin{definition}
	We specify the Kullback-Leibler divergence between two pmfs $P$ and $Q$ over the same domain by
	\begin{equation}
		D(P||Q) = \sum_{x \in \mathcal{X}} P(x) \log_2\left(\frac{P(x)}{Q(x)}\right) \enspace ,
	\end{equation}
	where we define $0 \log_2(0):=0$.
\end{definition}
\begin{definition}
	The mutual information between two random variables $X$ and $Y$ is specified by
	\begin{align}
		I(X;Y) &\coloneqq D(P_{XY}||P_X P_Y) \nonumber \\
		&= \sum_{(a,b) \in \supp(P_{XY})} P_{XY}(a,b) \log_2 \left( \frac{P_{XY}(a,b)}{P_X(a)P_Y(b)} \right) = H(Y) - H(Y|X) \enspace .
	\end{align}
	If the random variable $Y$ is continuous it holds that
	\begin{equation}
		I(X;Y) \coloneqq h(Y) - h(Y|X) \enspace .
	\end{equation}
\end{definition}

\begin{definition}
	The mutual information $I(X;Y|Z=c)$ between two random variables $X$ and $Y$ conditioned on a discrete random variable $Z$ taking the value $c$ is defined by
	\begin{equation}
		I(X;Y|Z=c) \defeq H(X|Z=c) - H(X|Y,Z=c)
	\end{equation}
	if $X$ is discrete and
	\begin{equation}
		I(X;Y|Z=c) \defeq h(X|Z=c) - h(X|Y,Z=c)
	\end{equation}
	if $X$ is continuous.
	Similar to Definition~\ref{def:disc_entropy} we define
	\begin{equation}
		I(X;Y|Z)\defeq \sum_{c\in \supp(P_Z)} P_Z(c)\, I(X;Y|Z=c) \enspace .
	\end{equation}
	For continuous $Z$ with pdf $f_Z$ we define
	\begin{equation}
		I(X;Y|Z) = \int_{c\in \supp(f_Z)} f_Z(c)\,  I(X;Y|Z=c) \diff c \enspace .
	\end{equation}
\end{definition}

\end{document}